\documentclass{article}

\usepackage{url}

\usepackage{times}
\usepackage{natbib}
\usepackage{soul}
\usepackage{xcolor}
\usepackage{amsmath, amssymb, amsthm}
\usepackage{verbatim}
\usepackage{mathrsfs}
\usepackage{subcaption}
\usepackage{tikz}
\usepgflibrary{arrows}
\usetikzlibrary{arrows.meta,positioning,shapes}
\usepackage{dsfont}
\usepackage{tcolorbox}
\usepackage{setspace}
\usepackage{xr}

\newtcbox{\inlinebox}[1][]{box align=base,
	nobeforeafter,
	size=small,
	left=0pt,
	right=0pt,
	boxsep=2pt,
	#1}

\newcommand{\an}{\mathrm{an}}
\newcommand{\ant}{\mathrm{an}_{\bullet\!}}
\newcommand{\antv}{\mathrm{an}_{\bullet\!}^{\scriptscriptstyle V}}

\newcommand{\pat}{\mathrm{pa}_{\bullet\!}}
\newcommand{\deta}{\mathrm{de}_{\bullet\!}}

\newcommand{\starrightarrow}{{\ {*\!\!\!}\rightarrow\ }}

\newcommand{\tailedrightarrow}{{\ {\bullet\!\!\!}\rightarrow\ }}
\newcommand{\tailedleftarrow}{{\ \leftarrow}{\!\bullet}\ }

\newtheorem{corollary}{Corollary}
\newtheorem{theorem}{Theorem}
\newtheorem{proposition}{Proposition}

\newtheorem{definition}{Definition}
\newtheorem{example}{Example}
\newtheorem{lemma}{Lemma}
\newtheorem{remark}{Remark}

\usepackage[]{graphicx}
\chardef\bslash=`\\ 

\makeatletter

\newcommand*{\addFileDependency}[1]{
	\typeout{(#1)}
	\@addtofilelist{#1}
	\IfFileExists{#1}{}{\typeout{No file #1.}}
}\makeatother

\begin{document}

\title{Time-dependent mediators in survival analysis: Graphical representation of causal assumptions}

\author{SØREN WENGEL MOGENSEN$^{1,\ast}$, ODD O. AALEN$^2$\\ AND 
SUSANNE 
STROHMAIER$^3$\\[20pt] {\itshape
	$^1$ Department of Automatic Control,}\\ {\itshape Lund University, Lund, 
	Sweden}\\[10pt] %
		$^2$ {\itshape Oslo Centre for Biostatistics and Epidemiology,}\\ 
		{\itshape Department 
		of Biostatistics, University of Oslo,
			Oslo, Norway}\\[10pt]  $^3$ {\itshape Center for Public Health, 
			Department 
			of Epidemiology,} \\ {\itshape Medical University of Vienna, 
			Vienna, Austria}}

\maketitle

\let\thefootnote\relax\footnotetext{$^\ast$ 
soren.wengel\_mogensen{@}control.lth.se}

\begin{abstract}
{We study time-dependent mediators in survival analysis using a treatment separation approach due to \cite{didelez2019defining} and based on earlier work by \cite{robins2010alternative}. This approach avoids nested counterfactuals and cross-world assumptions which are otherwise common in mediation analysis. The causal model of treatment, mediators, covariates, confounders and outcome is represented by causal directed acyclic graphs (DAGs). However, the DAGs tend to be very complex when we have measurements at a large number of time points. We therefore suggest using so-called \emph{rolled graphs} in which a node represents an entire coordinate process instead of a single random variable, leading us to far simpler graphical representations.

The rolled graphs are not necessarily acyclic; they can be analyzed by $\delta$-separation which is the appropriate graphical separation criterion in this class of graphs and analogous to $d$-separation. In particular, $\delta$-separation is a graphical tool for evaluating if the conditions of the mediation analysis are met or if unmeasured confounders influence the estimated effects.

We also state a mediational g-formula. This is similar to the approach in \cite{vansteelandt2019mediation} although that paper has a different conceptual basis. Finally, we apply this framework to a statistical model based on a Cox model with an added treatment effect.}{survival analysis; mediation; causal inference; graphical models; local independence graphs}
\end{abstract}

\section{Introduction}
Mediation analysis is an important aspect of causal inference. The purpose is to understand how causal effects are mediated through different causal pathways. As an example, assume there is a certain medication that has been shown to have an effect on a disease; one might then be interested in how this effect operates. It has, for instance, been documented that several treatments for high blood pressure reduce the risk of developing heart disease. An interesting question is to which extent the reduced blood pressure is the only, or main, reason why these medications work. For some medical treatments there are indications that they work through several pathways. The estimation of direct and indirect effects may throw some light on this issue.

Mediation analysis in a survival setting has been a challenge. One issue is that a patient may survive in one counterfactual setting, and not in the other, or they may be censored in one, and not in the other. An additional aspect is the need for understanding the development over time. The relevant mediators will typically be stochastic processes and one may want to model mediation as such. A carefully developed contribution was given by \cite{vansteelandt2019mediation}. The authors construct a counterfactual comparison where a hypothetical intervention on the mediator will set it to the level that would have been seen if the exposure had been different. If a person dies the mediator is set to a level that would occur if death had been prevented. These are so-called nested counterfactuals and they require an untestable cross-world assumption. This approach is subject to some dispute, see, e.g., \cite{didelez2019defining} and \cite{stensrud2022separable}. 

A simpler and potentially more intuitive procedure was developed by \cite{didelez2019defining}, based on earlier work by \cite{robins2010alternative}. Their approach allows a more straightforward translation from subject-matter questions to estimands. The idea is to look at different components of the treatment; one component describes how the treatment affects the outcome through the mediator, while another describes the other effects of treatment. The procedure was applied in \cite{aalen2020time} to a setting where only the last value of the mediator was used at any given time. A comparison of the method of \cite{vansteelandt2019mediation} with the approach used in \cite{aalen2020time} is given by \cite{tanner2022methods}.

Most of this paper models mediator and covariate processes as discrete-time processes. This corresponds to the way they are actually measured, e.g., from clinical examination of patients. The formalization of causal assumptions that we use are in the tradition of time-discrete causal directed acyclic graphs (DAGs) \citep{pearl09}. The underlying processes, e.g., blood pressure, typically evolve in continuous time. Mediation analysis in time-continuous models can also be done, but would be more complex in many settings. We give an example to illustrate how the graphical framework we propose may also be used in continuous-time models, even though most of this paper focuses on discrete-time stochastic processes.

Mediation analysis is natural in a clinical trial setting. An example is the study of cholesterol treatments (statins) as discussed by \cite{strohmaier15}; the outcome being coronary heart disease. As shown in that paper there is a major direct effect indicating that statins may have an important influence through other pathways than cholesterol. In such applications, it is crucial to be aware of the time delays that certainly would influence the processes. The effect on the outcome of a blood pressure or cholesterol medication, say, would typically not be immediate, but rather delayed due to the biological processes that are involved; see e.g. Figure 3 of the paper \cite{sprint2015randomized} where there is an initial period when no difference can be observed between survival curves.

In the present paper we give a further development of \cite{didelez2019defining} and \cite{aalen2020time}, achieving among other things the following:

\begin{itemize}
    \item We extend the mediation framework such that mediation is represented by an entire mediation process up to the present time.
    \item We formulate our mediation result using Granger non-causality (local independence) and we extend $\delta$-separation graphs to represent the assumptions of the mediation analysis.
    \item We prove a new $\delta$-separation Markov property in Granger causal graphs of time series with contemporaneous effects.
    \item We show how certain unmeasured confounder processes can be introduced without destroying the validity of the analysis.
    \item We discuss decomposing the treatment along pathways.
    \item We give an example of mediation analysis in a continuous-time model, illustrating that this graphical framework is applicable in both discrete- and continuous-time models.
\end{itemize}

Additional results and proofs are in Sections \ref{app:granger}-\ref{ssec:continuousTimeExample} in the supplementary materials.

\section{A causal survival model}

This section first describes the basic modelling framework that we will use. It then describes the treatment separation approach and a set of assumptions that are sufficient for mediation analysis.

\subsection{The basic model}
\label{basic model} 

We consider a treatment (or exposure), $A$, which may take one of two values, $a$ or $a^{\ast}$. We are mostly interested in a survival setting, hence, the outcome is the time of occurrence, $T$, of some event. Let $t_{0}=0<t_{1}<\ldots<t_{k}<\ldots$ be an increasing sequence of time points and let $T>t$ be an indicator function which is equal to 1 if the inequality is fulfilled and zero otherwise; we also write this as a counting process $N(t)$ and this outcome process could also be a more general stochastic process. The outcome, $N$, is influenced by two types of observable processes, namely mediator processes that carry the effect of treatment forward in time, and covariate processes that are not influenced by treatment. Let $\{M_{k}; k=0, 1,\ldots\}$ be a sequence of mediator values and $\{C_{k}; k=0, 1,\ldots\}$ a sequence of covariate values. Let $\overline{M}_{k}$ denote the vector of all $M_{j}$ for $j \leq k$, and similarly for the covariate process.  The values of mediator and covariate processes, $M_k$ and $C_k$, are only meaningful as long as the individual survives, that is, when $T > t_k$. We will assume that with probability 1 these processes assume some symbolic value indicating `not available' whenever $T \leq t_k$.

We assume that the data is time-ordered as
\begin{align*}
A, M_{0}, C_{0}, T>t_1,M_{1}, C_{1},  T>t_2, \ldots, T>t_k, M_{k}, C_{k}, T>t
\end{align*}
such that variables that are earlier in the above order are realised before variables that are later in the above order and $t_{k} < t \leq t_{k+1}$. We make the following assumption throughout the paper.

\begin{itemize}
\item {General causal assumption: for all $t$, the variables up to time $t$ are represented by a causal DAG as defined in \cite{pearl09} which is consistent with the above temporal order.} 
\end{itemize}

 An illustration of the DAG-model is given in Figure \ref{partial DAG} where we have included relevant variables up to time $t_2$. We assume at the initial stage that there are no unmeasured confounders, and we include such confounders in a later section.

 We assume that mediators and covariates are measured at fixed times, common for all individuals. The first mediator measurement is made at time 0 when the treatment starts. But, clearly, we cannot say anything about mediation before the second measurement is done at time $t_1$. Since the underlying framework is a counting process, the martingale assumptions require that we cannot go back in time. Hence analysis of mediation is done from time $t_1$ and forwards. It follows that mediation analysis with time-discrete measurements may tend to somewhat underestimate the real indirect effect through the mediator. A better estimate would be derived if the mediators and other variables were measured more frequently. However, a major area of application that we have in mind are clinical trials. In such studies the variables are measured quite frequently, so the discreteness in time would be a minor problem.

\begin{figure}
\centering
        \begin{subfigure}[t]{0.45\textwidth}
        \vskip 0pt
                \centering
                \subcaption{}
                \includegraphics[width=\linewidth]{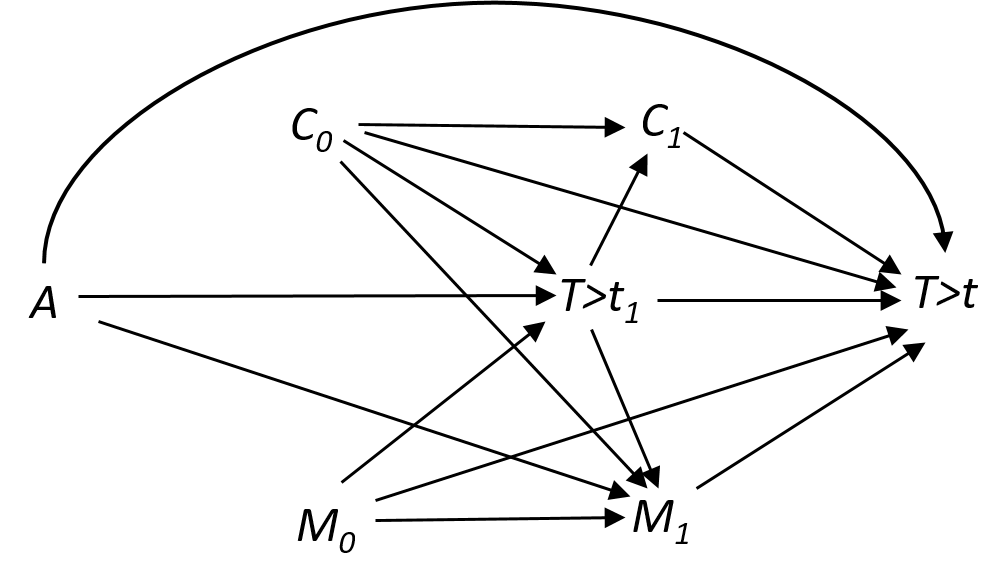}
                \label{partial DAG}
        \end{subfigure}\hfill
        \begin{subfigure}[t]{0.45\textwidth}
        \vskip 0pt
                \centering
                \subcaption{}
                \includegraphics[width=\linewidth]{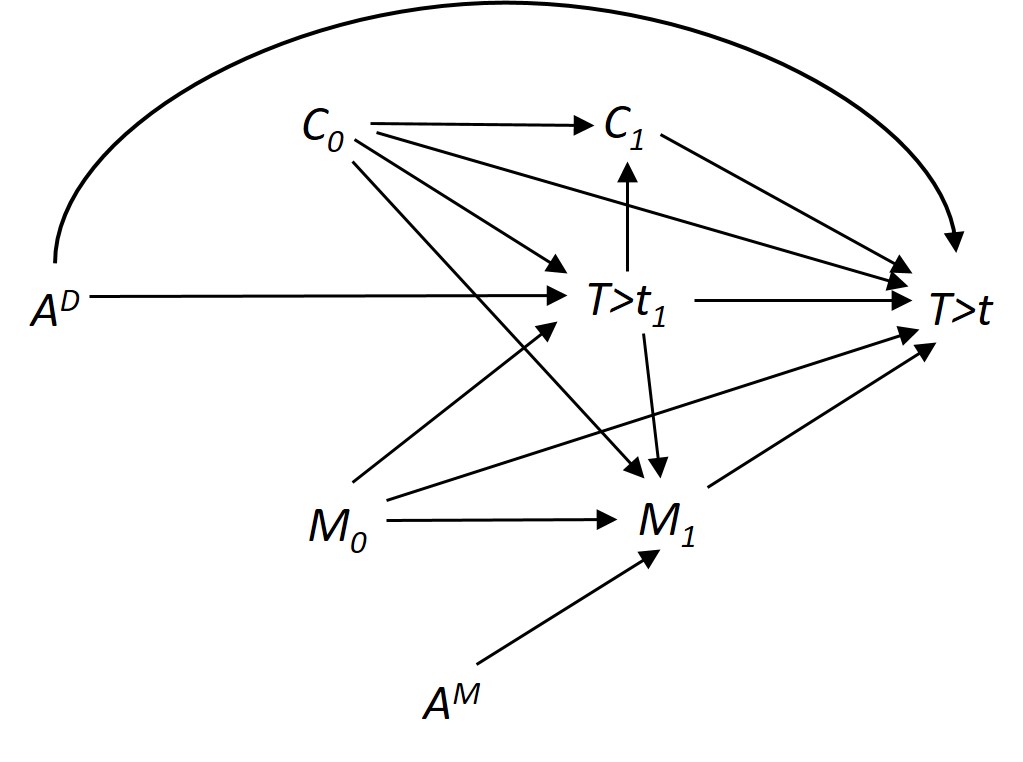}
                \label{DAG_confounders1}
        \end{subfigure}\hfill
        \caption{(\subref{partial DAG}) Graph for observed data: Causal DAG showing treatment $A$, mediators  $M_0$ and $M_1$, covariates $C_0$ and $C_1$ and the outcome $T>t$ for 
$t_{1} < t \leq t_{2}$. (\subref{DAG_confounders1}) Graph under the intervention $do(A^D = a, A^M = a^{\ast})$, see Subsection \ref{ssec:treatmentseparation}. Causal DAG showing treatment components $A^M$, $A^D$, mediators  $M_0$ and $M_1$, covariates $C_0$ and $C_1$ and the outcome $T>t$ for 
$t_{1} < t \leq t_{2}$.}
\end{figure}

\subsection{A treatment separation approach to mediation}
\label{ssec:treatmentseparation}

\cite{didelez2019defining} presented a new approach to mediation in survival analysis as a further development of work by \cite{robins2010alternative}. Instead of using natural direct and indirect effects, or variations of these, the emphasis is on understanding mediation through different aspects of treatment corresponding to different pathways. One imagines that the treatment, $A$, can be separated into two components; one, $A^M$, that operates through the mediator, and one, $A^D$, that operates `outside' the mediator, that is, directly. In observational data, there is a functional dependency between these variables as $A = A^D = A^M$ with probability 1. However, in some applications one may imagine hypothetical interventions which could break this functional dependence. When considering the effect of statins on heart disease, for example, one may imagine that a part of the treatment effect goes through a reduction in the cholesterol level, while another part is due to, e.g., changes in the risk of inflammation. As we will see, this approach avoids the nested counterfactuals and cross-world assumption of classical mediation theory.

Applications of the treatment separation approach have been given by \cite{aalen2020time} and it has been further extended to competing event settings \citep{stensrud2021generalized,stensrud2022separable} and truncation by death settings \citep{stensrud2022conditional}. The latter three references show how these ideas can clarify competing risks analysis by introducing a new estimand. Under the appropriate assumptions that will be outlined in later sections, the theoretical construct of treatment components allows us to identify direct and mediated effects, even though $A = A^D = A^M$ in the data we actually observe.

\subsection{Causal assumptions for mediation analysis}
\label{causmed}

Note that the DAG in Figure \ref{partial DAG} describes the observed data. We are now going a step further, describing the assumed relationship between the two treatment components (denoted $A^D$ and $A^M$)  and the observed data. We assume the following property \citep{didelez2019defining}:

\begin{itemize}
\item Property P1: The interventional distributions are such that for all $j$ :

\hspace{0.2cm} $ P( \; \cdot \; ; do(A^D = j, A^M = j)) = P( \; \cdot \; ; do(A = j)) $
\end{itemize}

Furthermore, we shall make the following causal assumptions for the mediation model. A0 should hold in the observational distribution while A1, A2 and A3 should hold in the interventional distribution corresponding to $do(A^D = a, A^M = a^{\ast})$. The point of making these assumptions is that they enable the important mediational g-formula in Section \ref{totpr}. 

\begin{itemize}

\item {A0: The treatment $A$ is randomised at time zero.}

\item {A1: For each time $t_k$ the mediator
$M_{k}$ is independent of the treatment component $A^D$  conditional on $T > t_k$,  $A^M$ and previous mediator and covariate values:

\hspace{0.2cm} $A^D \mathrel{\text{\scalebox{1.07}
{$\perp\mkern-10mu\perp$}}}M_{k} \thinspace \mid  \thinspace (T > t_k, A^M=a^{\ast}, \overline{M}_{k-1}, \overline{C}_{k-1})$}

\item {A2: For each $k$ and for each time $t$ satisfying $t_k < t \leq t_{k+1}$, the event $T > t$ is independent of the treatment component $A^M$  conditional on $T > t_k$, $A^D$ and previous mediator and covariate values:

\hspace{0cm} $A^M
\mathrel{\text{\scalebox{1.07}{$\perp\mkern-10mu\perp$}}}T > t \thinspace \mid  \thinspace ( T > t_k, A^D=a, \overline{M}_k,\overline{C}_{k}))$}

\item {A3: For each time $t_k$ the covariate
$C_{k}$ is independent of the treatment component $A^M$  conditional on $T > t_k$,  $A^D$ and previous mediator and covariate values:

\hspace{0.2cm} $A^M \mathrel{\text{\scalebox{1.07}
{$\perp\mkern-10mu\perp$}}}C_{k} \thinspace \mid  \thinspace (T > t_k, A^D=a, \overline{M}_{k}, \overline{C}_{k-1})$}

\end{itemize}

We can think of $A$ as occurring at time $t = 0$ and we assume that there are no contemporaneous effects from $A$, i.e., in the underlying causal graph there are no edges from $A$ to a variable at time $t = 0$. Note that Assumption A2 is written for $T>t$ where $t$ is an arbitrary time in the relevant interval and not necessarily one of the time points $t_k$. Such times $t$ can be included without any problem and this is done repeatedly below. 

Assumption A0 is convenient; however, it may be relaxed to allow observed confounders. Assumptions A1 and A2 are slight modifications of assumptions in \cite{didelez2019defining} and are made to ensure the validity of the treatment separation discussed above. For a more detailed discussion of what they mean, see \cite{didelez2019defining} and \cite{aalen2020time}. Assumption A3 is similar to assumption A1 and included here in order to give suitable restrictions on how the covariate process $C$ relates to the other processes as well as to the treatment.

One can use $d$-separation in the underlying causal DAG to argue for the validity of these assumptions as $d$-separation implies certain conditional independencies. For instance in Figure \ref{DAG_confounders1}, assumption A1 is implied by the graph since all paths from $A^D$ to $M_1$ are blocked by the conditioning variables in Assumption A1. However, it is clear that when we have many observed time points, the relevant DAGs are complex, and possibly rather unreadable. For this reason, we suggest in the next section a graphical framework (a \emph{rolled} graph) in which a node represents an entire coordinate process instead of a single random variable, leading us to far simpler graphical representations. Other graphical representations could also be used, e.g., DAGs with two lagged variables of each process, one lag representing the past of the process, and one lag representing the present. However, we use the rolled graphs as they also extend naturally to continuous-time models and connect nicely with known $\delta$-separation theory for local independence and Granger non-causality graphs.

\section{A stochastic process viewpoint}
\label{stochpro}

Local independence, or Granger non-causality in discrete-time models, describes how the evolution of a stochastic process depends on other stochastic processes. In the rolled graphs, we will use it analogously to how conditional independence is used in DAGs. In comparison with conditional independence, local independence has the advantage that there is a direction in the dependence relationship. The concept is often defined for time-continuous multivariate stochastic processes, using the Doob-Meyer decomposition which represents a submartingale as the sum of an increasing compensator and a martingale. A standard example is the decomposition of a counting process into the integral of the intensity process (the compensator) and a martingale. Local independence and local independence graphs were defined and studied by
\cite{schweder70,aalen87,didelez07,didelez08,aalen08,mogensen2018causal,mogensen2022graphical}, among others, for example in marked point processes \citep{didelez08} and multivariate Gaussian processes \citep{mogensen2022graphical}. An interesting study of graphical criteria in local independence graphs for nonparametric identification of causal effects is given by \cite{roysland2022graphical}. A stopped version of the local independence graph may be used in settings where observation is stopped, e.g., by death or censoring of an individual; see \cite{didelez08}.

In this section, we define \emph{Granger non-causality}, a time-discrete 
version of local independence, and we use this to lay out how cyclic graphs may 
represent the assumptions of our mediation analysis. In short, the causal 
assumptions we make are on an underlying causal DAG-model as in the previous 
section. From this DAG we construct a `rolled', and possibly cyclic, graph in 
which each node represents an entire coordinate process (Subsection 
\ref{relationship}). We show that the conditional independencies represented by 
the DAG implies local independencies that can be read off from the rolled graph 
(Subsection \ref{relationship}) using $\delta$-separation. Finally, we show how 
this can be used to represent mediation analysis assumptions, also in the 
presence of unobserved confounding (Section \ref{unmeasured confounders}).

\subsection{The time-discrete setting}
\label{time-discrete}

The distinction between data observed in continuous time and discrete time is ubiquitous in survival analysis. In causal modelling, it would seem most appropriate to consider covariates, mediators, outcomes and confounders to be evolving in continuous time. However, most often they would only be measured at discrete points in time. For instance, blood pressure is typically measured at visits to the doctor while clearly the blood pressure changes between visits. This is a limitation if the true model is continuous-time and the sampling frequency is not sufficiently high. We will, as is common, assume that the sampling frequency is sufficiently high and build a causal model in discrete time.

Time-discrete local independence and related graphical modelling have been discussed by \cite{eichler2010granger,eichler2012causal}, \cite{eichler2013causal}, \cite{didelez08} and the supplementary material of \cite{mogensen2020markov}, often using the concept of Granger non-causality which is a time-discrete analogue to local independence. Below we define Granger non-causality in discrete-time stochastic processes, that is, in time series, and in our case we need to allow the inclusion of baseline variables.

Let $(Z_t)_{t\geq 0}$ be a multivariate time series such that $Z_t$ is a random vector with entries indexed by a finite set $V$. The baseline variables are a random vector denoted by $B$ and we let $W$ be an index set of the baseline variables. We let $X = \{ X_t\}_{t\geq 0}$ be the multivariate stochastic process such that $X_0$ is the concatenation of $Z_0$ and $B$ and such that $X_t = Z_t$ for $t>0$. We say that $X$ is a time series with baseline variables and index sets $(V,W)$. Note that, in the presence of baseline variables, random vectors $X_0$ and $X_t$, $t>0$, are of different length. We define $\overline{X}_t^A= \{X_t^A;s\leq t\}$ for any subset $A$ of $V\cup W$ where $X_t^A$ is the subvector of $X_t$ indexed by $A$. We use the following extended definition of Granger non-causality which combines stochastic processes and random variables.

\begin{definition}[Granger non-causality with baseline variables]
Let $X$ be a time series with baseline variables and index sets $(V,W)$. Let $A$, $B$ and $C$ be disjoint subsets of $V \cup W$ such that $B\subseteq V$. We say that $X^A$ is Granger non-causal for $X^B$ given $X^{C}$ if for all $t\geq 1$ we have that $\overline{X}^A_{t-1}$ and $X_t^B$ are conditionally independent given $\overline{X}^{B\cup C}_{t-1}$. We say that $X^A$ is Granger non-causal for $X^B$ given $X^C$ until time ${t'}$ if this holds for all $t=1,2\ldots, {t'}$. 
\label{def:grangercausality}
\end{definition}

When $X^A$ is Granger non-causal for $X^B$ given $X^C$, we shall also say that $X^B$ is locally independent of $X^A$ given $X^C$ and denote it by $X^A \not\rightarrow X^B \mid X^C$ or $A \not\rightarrow B \mid C$. We provide an equivalent definition of Granger non-causality in Appendix \ref{app:granger}. The alternative definition is also widespread in the literature, though the above is more closely aligned with the notion of graphical separation that we will use.

We will allow contemporaneous effects in the underlying causal DAG and one may also consider contemporaneous (in)dependencies in addition to the above notion of independence, see, e.g., \cite{eichler2010granger}. However, for our purposes the above definition suffices.

\subsection{Local independence and causality}

A general way to define causal concepts is based on the idea of interventions, following for instance \cite{pearl09}. The connection between local independence and intervention-based causality has been studied by \cite{roysland12}, \cite{didelez2015causal} and \cite{mogensen2020causal}. For a discussion of interventions in the setting of Granger non-causality, see \citet{eichler2010granger,eichler2012causal}. These approaches all employ the idea that a \emph{causal} model should satisfy a certain \emph{structural stability} under interventions. One imagines that a system can be described as a set of modules, and that  intervention takes place on some of the modules while the modules not intervened upon and the underlying structure remain stable. Intervention could, e.g., take place on an intensity process of a counting process, or one could intervene on one coordinate process in a time series. The remaining structure stays unchanged, however, the exact meaning of this depends on the model class.

In this paper, the causal interpretation of the local independence graphs of discrete-time models is particularly simple as we make the causal assumption on the variable level by interpreting the underlying DAG causally. The local independence graphs are therefore to be understood only as convenient graphical representations. In Appendix \ref{ssec:continuousTimeExample} in the supplementary materials, we use a continuous-time model and this requires a different causal interpretation. We describe this in the same appendix.

\subsection{Local independence and $\delta$-separation}
\label{delta-separation}

Local independence may be represented graphically using the very useful concept of $\delta$-separation; see \cite{didelez08} for a practical example. $\delta$-separation is a graphical criterion which implies local independence using a so-called \emph{global Markov property}. This is analogous to how $d$-separation in directed acyclic graphs implies conditional independence \citep{pearl09}. While $d$-separation and conditional independence are \emph{symmetric} (e.g., $X \mathrel{\text{\scalebox{1.07}
{$\perp\mkern-10mu\perp$}}} Y \thinspace \mid  \thinspace Z$ implies $Y \mathrel{\text{\scalebox{1.07}
{$\perp\mkern-10mu\perp$}}} X \thinspace \mid  \thinspace Z$) this is not the case for $\delta$-separation and local independence.

Markov property results for stationary time series are given by \cite{eichler2007granger}. We shall need results for non-stationary time-discrete processes and these are described in Section \ref{relationship} which extends results in the supplementary material of \cite{mogensen2020markov} to allow baseline variables and contemporaneous effects. These results imply that we can use $\delta$-separation to assess if unmeasured confounders can be ignored in the mediation analysis.
 
Before introducing $\delta$-separation, we will define various graph-theoretic notions. Formally, a graph $D = (V,E)$ consists of a set of \emph{nodes}, $V$, and a set of edges, $E$. Each edge is between a pair of nodes (not necessarily distinct). In this paper, we will mostly consider \emph{directed graphs} (DGs) in which every edge is \emph{directed}, that is, $i\rightarrow j$ or $i \leftarrow j$, $i,j\in V$. We assume that between any pair, $(i,j)$, of nodes there is a most one edge pointing from $i$ to $j$, $i\rightarrow j$, however, both $i\rightarrow j$ and $i \leftarrow j$ may be present in the graph. A \emph{walk}, $\pi$, is an alternating sequence of nodes and edges, $\langle i_0, e_1 , i_1, e_2, \ldots, e_n, i_n\rangle$ such that $i_k \in V$, $e_l \in E$, and $e_l$ is between $i_{l-1}$ and $i_l$. A \emph{path} is a walk such that no node is repeated. For $0<k<n$, we say that $i_k$ is a \emph{collider} on $\pi$ if $k_i$ and $e_{k+1}$ both have \emph{heads} at $i_k$, i.e., $\rightarrow i_k\leftarrow$, and otherwise we say that $i_k$ is a \emph{noncollider}. A \emph{directed path} is a path such that all edges point in the same direction. A \emph{directed cycle} is a self-loop, $i \rightarrow i$, or a directed path from $i$ to $j$ along with the directed edge $j\rightarrow i$. A \emph{directed acyclic graph} (DAG) is a DG with no directed cycles and we say that a DG with no directed cycles is \emph{acyclic}. If $i \rightarrow j$, then we say that $i$ is a \emph{parent} of $j$ and that $j$ is a \emph{child} of $i$. If there is a directed path from $i$ to $j$, and $i \neq j$, then we say that $i$ is an \emph{ancestor} of $j$. We denote the set of ancestors of a node by $\an(i)$ and we define $\an(A) = \cup_{i\in A} \an(i)$ for $A\subseteq V$. Finally, we let $D_A = (A, E_A)$ denote the \emph{subgraph induced by $A$} where $E_A$ is the subset of edges in $E$ which are between two nodes in $A$.

In general each node in the node set of a DG will represent a coordinate process or a baseline variable. For $B \subseteq V$, we define an auxiliary graph $D^B$ by deleting all
directed edges starting from $B$, that is, edges $i\rightarrow j$ such that 
$i\in B$.

\begin{definition}[$\delta$-separation, \cite{didelez08}] Let $A$, $B$ and $C$ be disjoint subsets of $V$. We say that $B$ is $\delta$-separated from $A$ given $C$ if every path between $i \in A$ and $j \in B$ in the graph $D^B$ either contains a noncollider $k$ such that $k \in C$ or a collider $k$ such that $k \notin \an_D(C) \cup C$. We denote $\delta$-separation by $A \not\rightarrow_\delta B \mid C$.
\label{def:deltaSep}
\end{definition}

\cite{mogensen2020markov} introduced the concept of $\mu$-separation which is a generalization of $\delta$-separation.
We will often drop the subscript $D$ in $\an_D(C)$ and we let $\an^\texttt{+}(C)$ denote $\an(C) \cup C$. For completeness, we also state the classical definition of $d$-separation. The concept of $\delta$-separation is clearly similar to $d$-separation and in the next section we give an intuitive explanation of why edges are removed to obtain $D^B$ in the definition of $\delta$-separation which is different from $d$-separation. In this paper, $\delta$-separation is applied to graphs that may be cyclic while $d$-separation is only applied to acyclic graphs.

\begin{definition}[$d$-separation, \cite{pearl09}]
    Let $A$, $B$ and $C$ be disjoint subsets of $V$. We say that $A$ and $B$ are $d$-separated given $C$ if every path between $i \in A$ and $j \in B$ in $D$ either contains a noncollider $k$ such that $k \in C$ or a collider $k$ such that $k \notin \an(C) \cup C$.
\end{definition}

\subsection{Relationship between $\delta$-separation and $d$-separation graphs}
\label{relationship}

When we use graphs to represent a time series, nodes can represent either random variables or entire processes. In this subsection, we describe the connection between these two representations. As above, we consider a time series, $X$, with baseline variables and index sets $(V,W)$ where $V$ is an index set of the coordinate processes and $W$ is an index set of baseline variables. We assume throughout that $V$ and $W$ are disjoint. We represent this process using a DG, $D = (V\cup W,E)$. In this representation, each node $i\in V\cup W$ represents either an entire stochastic process, $X^i$, $i\in W$, or a baseline variable, $i\in W$, and we will say that such a DG is a \emph{rolled graph}. We will also use DAGs in which each node, $\nu_s^i$, represents a single random variable in the time series, $X_s^i$, $i\in V\cup W$, and we will say that such a DAG is an \emph{unrolled graph}. We define \emph{rolling} and \emph{unrolling} operations to translate between these two graphical representations and we will later extend these definitions to also allow contemporaneous effects. Several authors have used similar notions of rolled and unrolled graphical representations of stochastic processes, see, e.g., \cite{danks2013, sokol2014, danks2016}. The term \emph{unfolded graph} has also been used in the literature.

\begin{figure}
	    \newcommand\xx{2}
	    \newcommand\yy{2}
\begin{minipage}{0.2\textwidth}
				\begin{subfigure}{\textwidth}
					\centering
					\begin{tikzpicture}[scale=0.7]
					\tikzset{vertex/.style = {shape=circle,draw,minimum 
					size=1.5em, 
							inner 
							sep = 0pt}}
					\tikzset{edge/.style = {->,> = latex', thick}}
					\tikzset{edgebi/.style = {<->,> = latex', thick}}
					\tikzset{every loop/.style={min distance=8mm, looseness=5}}
					\tikzset{vertexFac/.style = {shape=rectangle,draw,minimum 
							size=1.5em, 
							inner sep = 0pt}}
					
					\node[vertex] (A) at  (0,0) {$P$}; 
					\node[vertex] (C) at  (0,-2) {$Q$}; 
					\node[vertex] (M) at  (2,0) {$R$}; 
					\node[vertex] (N) at  (2,-2) {$S$}; 
					
					\node (a) at (-0.5,1.75) {\textbf{A}};
					
					\draw[edge] (A) to (N);
					\draw[edge] (C) to (M);
					\draw[edge] (N) to (C);
					\draw[edge] (M) to (N);
					
					\end{tikzpicture}
				\end{subfigure}
    \end{minipage}\hfill\vrule\hspace{.04\textwidth}\hfill
    \begin{minipage}{.3\textwidth}
				\begin{subfigure}{0.3\textwidth}
					\centering
					\begin{tikzpicture}[scale=0.7]
					\tikzset{vertex/.style = {shape=circle,draw,minimum 
					size=1.5em, 
							inner 
							sep = 0pt}}
					\tikzset{edge/.style = {->,> = latex', thick}}
					\tikzset{edgebi/.style = {<->,> = latex', thick}}
					\tikzset{every loop/.style={min distance=8mm, looseness=5}}
					\tikzset{vertexFac/.style = {shape=rectangle,draw,minimum 
							size=1.5em, 
							inner sep = 0pt}}
					
					\node[vertex] (m0) at  (0,0) {$\nu_0^R$}; 
					\node[vertex] (c0) at  (0,-1*\yy) {$\nu_0^Q$}; 
					\node[vertex] (n0) at  (0,-2*\yy) {$\nu_0^S$}; 
					\node[vertex] (a0) at  (0,-3*\yy) {$\nu_0^P$}; 
					
					\node[vertex] (m1) at  (\xx,0) {$\nu_1^R$};
					\node[vertex] (c1) at  (\xx,-1*\yy) {$\nu_1^Q$};
					\node[vertex] (n1) at  (\xx,-2*\yy) {$\nu_1^S$};
					
					\node[vertex] (m2) at  (2*\xx,0) {$\nu_2^R$};
					\node[vertex] (c2) at  (2*\xx,-1*\yy) {$\nu_2^Q$};
					\node[vertex] (n2) at  (2*\xx,-2*\yy) {$\nu_2^S$};

					\node at (-0.5,1.75) {\textbf{B}};
					
					
					\draw[edge] (a0) to (n1);
					\draw[edge] (a0) to (n2);
					
					\draw[edge] (n0) to (n1);
					\draw[edge] (n1) to (n2);
					
					\draw[edge] (c0) to (c1);
					\draw[edge] (c1) to (c2);
					\draw[edge] (c0) to (m1);
					\draw[edge] (c0) to (m2);
					\draw[edge] (c1) to (m2);
					\draw[edge] (n0) to (c1);
					\draw[edge] (n0) to (c2);
					\draw[edge] (n1) to (c2);
					
					\draw[edge] (m0) to (m1);
					\draw[edge] (m1) to (m2);
					\draw[edge] (m0) to (n1);
					\draw[edge, bend right = 30] (m0) to (n2);
					\draw[edge] (m1) to (n2);
					
					\end{tikzpicture}
				\end{subfigure}
    \end{minipage}\hfill\vrule\hspace{.04\textwidth}\hfill
    \begin{minipage}{.3\textwidth}
    				\begin{subfigure}{0.3\textwidth}
					\centering
					\begin{tikzpicture}[scale=0.7]
					\tikzset{vertex/.style = {shape=circle,draw,minimum 
					size=1.5em, 
							inner 
							sep = 0pt}}
					\tikzset{edge/.style = {->,> = latex', thick}}
					\tikzset{edgebi/.style = {<->,> = latex', thick}}
					\tikzset{every loop/.style={min distance=8mm, looseness=5}}
					\tikzset{vertexFac/.style = {shape=rectangle,draw,minimum 
							size=1.5em, 
							inner sep = 0pt}}
					
					\node[vertex] (m0) at  (0,0) {$\nu_0^R$}; 
					\node[vertex] (c0) at  (0,-1*\yy) {$\nu_0^Q$}; 
					\node[vertex] (n0) at  (0,-2*\yy) {$\nu_0^S$}; 
					\node[vertex] (a0) at  (0,-3*\yy) {$\nu_0^P$}; 
					
					\node[vertex] (m1) at  (\xx,0) {$\nu_1^R$};
					\node[vertex] (c1) at  (\xx,-1*\yy) {$\nu_1^Q$};
					\node[vertex] (n1) at  (\xx,-2*\yy) {$\nu_1^S$};
					
					\node[vertex] (m2) at  (2*\xx,0) {$\nu_2^R$};
					\node[vertex] (c2) at  (2*\xx,-1*\yy) {$\nu_2^Q$};
					\node[vertex] (n2) at  (2*\xx,-2*\yy) {$\nu_2^S$};

					\node at (-0.5,1.75) {\textbf{C}};
					
					
					\draw[edge] (a0) to (n1);
					\draw[edge] (a0) to (n2);
					
					\draw[edge] (n0) to (n1);
					\draw[edge] (n1) to (n2);
					
					\draw[edge] (c0) to (c1);
					\draw[edge] (c1) to (c2);
					\draw[edge] (c0) to (m2);
					\draw[edge] (c1) to (m2);
					\draw[edge] (n0) to (c2);
					\draw[edge] (n1) to (c2);
					
					\draw[edge] (m0) to (m1);
					\draw[edge] (m1) to (m2);
					\draw[edge] (m0) to (n1);
					\draw[edge, bend right = 30] (m0) to (n2);
					\draw[edge] (m1) to (n2);
				\end{tikzpicture}
				\end{subfigure}
    \end{minipage}\vspace{.5cm}
    \begin{minipage}{.45\textwidth}
				\begin{subfigure}{\textwidth}
					\centering
					\begin{tikzpicture}[scale=0.7]
					\tikzset{vertex/.style = {shape=circle,draw,minimum 
					size=1.5em, 
							inner 
							sep = 0pt}}
					\tikzset{edge/.style = {->,> = latex', thick}}
					\tikzset{tedge/.style = {Circle->,> = latex', thick, black!50!green}}
					\tikzset{edgebi/.style = {<->,> = latex', thick}}
					\tikzset{every loop/.style={min distance=8mm, looseness=5}}
					\tikzset{vertexFac/.style = {shape=rectangle,draw,minimum 
							size=1.5em, 
							inner sep = 0pt}}
					
					\node[vertex] (A) at  (0,0) {$P$}; 
					\node[vertex] (C) at  (0,-2) {$Q$}; 
					\node[vertex] (M) at  (2,0) {$R$}; 
					\node[vertex] (N) at  (2,-2) {$S$}; 
					
					\node at (-0.5,1.75) {\textbf{D}};
					
					
					\draw[edge] (A) to (N);
					\draw[tedge] (C) to (M);
					\draw[tedge] (N) to (C);
					\draw[edge] (M) to (N);
					
					\end{tikzpicture}
				\end{subfigure}
\end{minipage}\hfill\vrule\hspace{.04\textwidth}\hfill
    \begin{minipage}{.45\textwidth}
    				\begin{subfigure}{\textwidth}
					\centering
					\begin{tikzpicture}[scale=0.7]
					\tikzset{vertex/.style = {shape=circle,draw,minimum 
					size=1.5em, 
							inner 
							sep = 0pt}}
					\tikzset{edge/.style = {->,> = latex', thick}}
					\tikzset{edgebi/.style = {<->,> = latex', thick}}
					\tikzset{every loop/.style={min distance=8mm, looseness=5}}
					\tikzset{vertexFac/.style = {shape=rectangle,draw,minimum 
							size=1.5em, 
							inner sep = 0pt}}
					
					\node[vertex] (m0) at  (0,0) {$\nu_0^R$}; 
					\node[vertex] (c0) at  (0,-1*\yy) {$\nu_0^Q$}; 
					\node[vertex] (n0) at  (0,-2*\yy) {$\nu_0^S$}; 
					\node[vertex] (a0) at  (0,-3*\yy) {$\nu_0^P$}; 
					
					\node[vertex] (m1) at  (\xx,0) {$\nu_1^R$};
					\node[vertex] (c1) at  (\xx,-1*\yy) {$\nu_1^Q$};
					\node[vertex] (n1) at  (\xx,-2*\yy) {$\nu_1^S$};
					
					\node[vertex] (m2) at  (2*\xx,0) {$\nu_2^R$};
					\node[vertex] (c2) at  (2*\xx,-1*\yy) {$\nu_2^Q$};
					\node[vertex] (n2) at  (2*\xx,-2*\yy) {$\nu_2^S$};

					\node at (-0.5,1.75) {\textbf{E}};
					
					
					\draw[edge] (a0) to (n1);
					\draw[edge] (a0) to (n2);
					
					\draw[edge] (n0) to (n1);
					\draw[edge] (n1) to (n2);
					
					\draw[edge] (c0) to (c1);
					\draw[edge] (c1) to (c2);
					\draw[edge] (c0) to (m1);
					\draw[edge] (c0) to (m2);
					\draw[edge] (c1) to (m2);
					\draw[edge] (n0) to (c1);
					\draw[edge] (n0) to (c2);
					\draw[edge] (n1) to (c2);
					
					\draw[edge] (m0) to (m1);
					\draw[edge] (m1) to (m2);
					\draw[edge] (m0) to (n1);
					\draw[edge, bend right = 30] (m0) to (n2);
					\draw[edge] (m1) to (n2);

					\draw[edge] (n0) to (c0);
					\draw[edge] (c0) to (m0);
					\draw[edge] (n1) to (c1);
					\draw[edge] (c1) to (m1);
					\draw[edge] (n2) to (c2);
					\draw[edge] (c2) to (m2);
				\end{tikzpicture}
				\end{subfigure}
    \end{minipage}
    				\caption{\label{fig:roll} \inlinebox{\textbf{A}} An example of a DG, $D$, on nodes $V \cup W = \{ P,Q,R,S\}$ such that $W = \{P\}$ is the set of nodes corresponding to baseline variables. \inlinebox{\textbf{B}} A DAG which is the unrolled version of $D$ on two lags. \inlinebox{\textbf{C}} Another DAG which is a subgraph of graph \textbf{B} (edges $\nu_0^Q \rightarrow \nu_1^R$ and $\nu_0^S \rightarrow \nu_1^Q$ are present in one and absent in the other). One obtains graph \textbf{A} when rolling either graph \textbf{B} or graph \textbf{C} which illustrates that the rolling operation is not injective without further assumptions on the DAGs. Edges $\nu_0^F \rightarrow \nu_2^F$, $F\in \{R,Q,S \}$ are omitted in graphs \textbf{B}, \textbf{C} and \textbf{E}. \inlinebox{\textbf{D}} An extended local independence graph with contemporaneous effects (see Section \ref{deltasepContemp}). \inlinebox{\textbf{E}} Unrolled version of \textbf{D}. In \textbf{A} and \textbf{D}, $R$ is $\delta$-separated from $S$ by $Q$, corresponding to the fact that there is no $d$-connecting walk between $\nu_2^R$ and $\{\nu_0^S,\nu_1^S\}$ given $\{\nu_0^Q,\nu_1^Q,\nu_0^R,\nu_1^R\} $ in \textbf{B}. On the other hand, there is a $d$-connecting walk between $\nu_2^R$ and $\{\nu_0^S,\nu_1^S\}$ given $\{\nu_0^Q,\nu_1^Q,\nu_0^R,\nu_1^R\} $ in \textbf{E}. This illustrates that Proposition \ref{prop:markov} is not sufficient when also allowing contemporaneous effects.}
\end{figure}
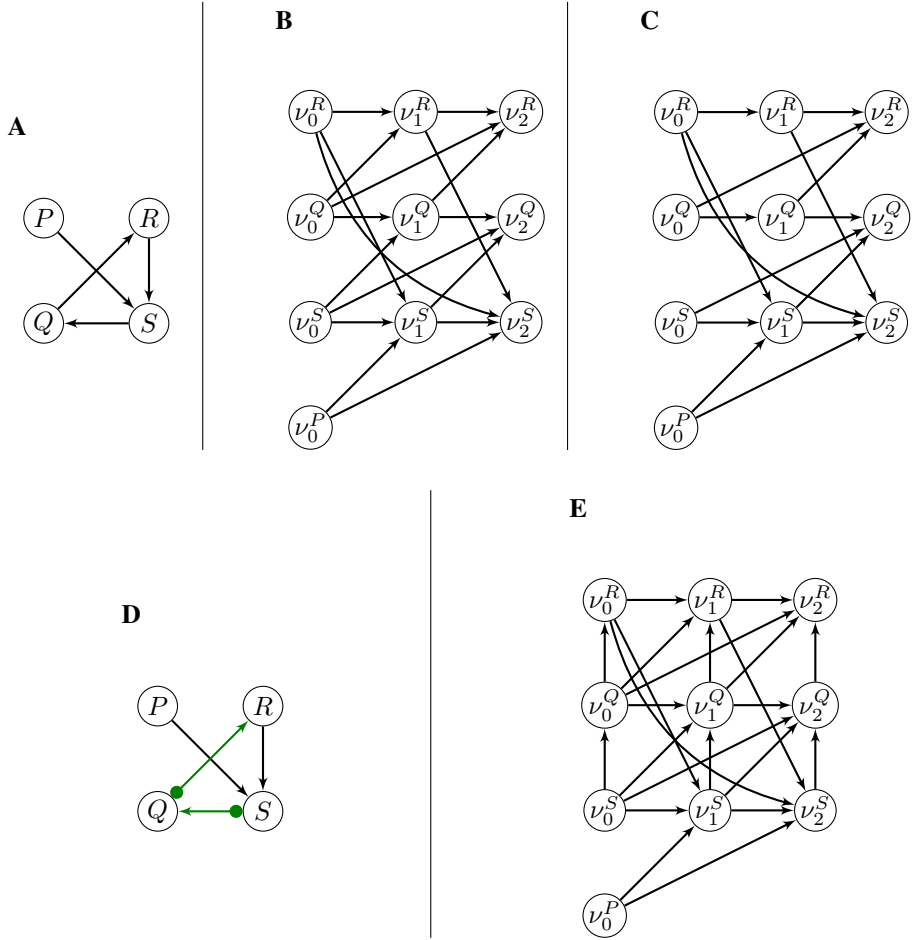

\begin{definition}[Unrolling]
\label{def:unrolling}
    Let $D$ be a directed graph on nodes $V \cup W$ such that the nodes in $W$ correspond to baseline nodes. The \emph{unrolled version} of $D$ on ${t'}$ lags, ${t'}\geq 1$, $D_{t'}$, is the DAG on nodes $ V^{t'} \cup W^{t'}$, 
    
    $$V^{t'} = \cup_{i\in V} \{ \nu_0^i, \ldots, \nu_{t'}^i \}, \ \ \  W^{t'} = \{\nu_0^i, i\in W \}$$
    
    \noindent such that $\nu_s^i \rightarrow \nu_t^j$ in $D_{t'}$ if $s<t$ and $i\rightarrow j$ in $D$.
\end{definition}

From the above definition, we see that when applying the unrolling operation, nodes corresponding to baseline variables are not `unrolled'. We also define a transformation in the opposite direction to obtain a `rolled' graph from a DAG. The above definition applies to any DG, however, for simplicity we will only formulate the reverse operation for DAGs with a certain structure.

\begin{definition}[Rolling]
\label{def:rolling}
    Let $D_{t'}$ be a DAG on nodes $V^{t'} \cup W^{t'}$ such that $V^{t'} = \cup_{i\in V} \{ \nu_0^i, \ldots, \nu_{t'}^i \}$ and $W^{t'} = \{ \nu_0^i, i\in W \}$ for disjoint sets $V$ and $W$. The \emph{rolled version} of $D_{t'}$ is the directed graph, $D$, on nodes $V \cup W$ such that for $i,j \in V \cup W$ we have that $i\rightarrow j$ if there exists $s<t$ such that $\nu_s^i \rightarrow \nu_t^j$ in $D_{t'}$.
\end{definition}

The operation of rolling is not injective (see Figure \ref{fig:roll} for an example to illustrate this).
We use uppercase Latin letters ($A,B,C,\ldots$) to denote subsets of nodes in rolled graphs, that is, subsets of $V\cup W$. In unrolled graphs, we use two ways to denote sets of nodes. First, lowercase Latin letters, e.g., $a$, to represent a general subset of $V^{t'} \cup W^{t'}$. We let $x^a$ denote the corresponding set of random variables, $a\subseteq V^{t'} \cup W^{t'}$. Second, for $A\subseteq V\cup W$ we let $\nu_t^A$ denote the nodes $\{\nu_t^i : i \in A \}\subseteq V^{t'} \cup W^{t'}$ and we let $\bar{\nu}_t^A$ denote $\{\nu_s^i : i \in A, s\leq t \}\subseteq V^{t'} \cup W^{t'}$. The corresponding sets of random variables are denoted by $X_t^A$ and $\bar{X}_t^A$, respectively. If $A \subseteq V\cup W$ consists of baseline variables only, that is, $A\subseteq W$, then $\bar{\nu}_t^A = \{\nu_0^i: i\in A\}$ for $t\geq 0$.

If a graph $D$ is a rolled version of $D_{t'}$, or if $D_{t'}$ is an unrolled version of $D$, then we say that the pair $(D,D_{t'})$ is \emph{proper}. The rolling operation is not bijective, but the two graphs in a proper pair clearly represent similar restrictions on the underlying causal DAG. When $(D,D_{t'})$ is a proper pair, we let $V^{t'} \cup W^{t'}$ denote the node set of $D_{t'}$. We say that an edge in $D_{t'}$, $\nu_s^i \rightarrow \nu_t^j$ is \emph{contemporaneous} if $s = t$. When $D_{t'}$ denotes the causal graph on ${t'}$ lags, we assume throughout that the standard $d$-separation Markov property holds in $D_{t'}$ for all ${t'}$. That is, for all $a,b,c \subseteq (V^{t'} \cup W^{t'})$, if $a$ and $b$ are $d$-separated by $c$ in $D_{t'}$ then $x^a$ and $x^b$ are conditionally independent given $x^c$.

The following proposition extends a result in the supplementary material of \cite{mogensen2020markov} to allow baseline variables.

\begin{proposition}[$\delta$-separation Markov property]
\label{prop:markov}
    Assume that $(D,D_{t'})$ is proper, and that $D_{t'}$ has no contemporaneous edges. Assume furthermore that the $d$-separation Markov property holds in $D_{t'}$. Let $A,B,C \subseteq V\cup W$ be disjoint. If $B$ is $\delta$-separated from $A$ given $C$ in $D$, then $X^B$ is locally independent of $X^A$ given $X^C$ until time ${t'}$.
\end{proposition}

The above result states that whenever $d$-separation in the unrolled graph, $D_{t'}$, implies conditional independence, then $\delta$-separation in the rolled graph, $D$, implies local independence (i.e., Granger non-causality). That is, the so-called global Markov property in the unrolled graph implies a global Markov property in the rolled graph.  This is a useful result since the $\delta$-separation results described can be used to represent the assumptions we need for the mediation analysis (Section \ref{singmed}), with some adjustments. There is a connection to results for stationary time series by \cite{eichler2007granger}, but here we allow non-stationarity which is natural in our setting.

In the above proposition, contemporaneous edges are not allowed and therefore there are no edges between baseline variables. We remove this restriction in the next section.

\begin{remark}
The concept of $\delta$-separation is relevant due to its relation with local independence. However, at a first glance it may not be clear why edges pointed away from $B$ are removed when deciding if $B$ is $\delta$-separated from $A$ given $C$. To understand this, one may look at a local independence graph and its unrolled graph, e.g., graphs \textbf{A} and \textbf{B} in Figure \ref{fig:roll}. $R$ is $\delta$-separated from $S$ given $\{Q,P\}$ as the edge $R \rightarrow S$ is removed to obtain $D^{\{R\}}$. In the unrolled graph, this edge corresponds to edges $\nu_s^R\rightarrow \nu_t^S$ for $s<t$. In terms of $d$-separation in the unrolled graph, this means that paths between $\nu_2^R$ and $\{\nu_0^S,\nu_1^S \}$ which include an edge $\nu_s^R\rightarrow \nu_t^S$, $s<t$, would always contain a noncollider in the past of $R$ which is blocked as the past of $R$ is tacitly conditioned upon (see Definition \ref{def:grangercausality}).
\end{remark}

\subsection{$\delta$-separation and contemporaneous effects}
\label{deltasepContemp}

Local independence graphs and $\delta$-separation as introduced disallow causal effects within a time slice. Contemporaneous effects may change the local independencies represented by a graph which means that to include contemporaneous effects we must extend the $\delta$-separation framework. For simplicity we assume that if $\nu_t^i \rightarrow \nu_t^j$ is in the causal graph, then also $\nu_s^i \rightarrow \nu_u^j$ for some $s < u$.

For this extension, we use graphs with both \emph{directed edges}, $\rightarrow$, and \emph{tailed directed edge}, $\tailedrightarrow$, and we will say that these are \emph{tailed directed graphs} or \emph{extended local independence graphs}. In visualizations, tailed directed edges are green. The edge $\rightarrow$ will have the same meaning as above while $\tailedrightarrow$ will represent the fact that there could be both contemporaneous and lagged effects. For this reason, if $i \tailedrightarrow j$ is in the graph, there is no reason to include $i \rightarrow j$ as well and we will omit the latter edge. We use $i\starrightarrow j$ to denote that $i\rightarrow j$ or $i\tailedrightarrow j$. We will assume that only baseline variables have edges into baseline variables, i.e., if $i\starrightarrow j$ and $j\in W$, then $i\tailedrightarrow j$ and $i\in W$. The rolling and unrolling operations to translate between the underlying causal DAG and the corresponding extended local independence graph only need minor adjustments.

\begin{definition}[Unrolling, contemporaneous effects]
\label{def:unrollingExt}
    Let $D$ be a tailed directed graph on nodes $V \cup W$ such that the nodes in $W$ correspond to baseline variables, and let ${t'}\geq 1$. The \emph{unrolled version} of $D$ on ${t'}$ lags, $D_{t'}$, is the DAG on nodes $ V^{t'} \cup W^{t'}$, 
    
    $$V^{t'} = \cup_{i\in V} \{ \nu_0^i, \ldots, \nu_{t'}^i \}, \ \ \  W^{t'} = \{\nu_0^i, i\in W \}$$
    
    \noindent such that $\nu_s^i \rightarrow \nu_t^j$ in $D_{t'}$ if $s<t$ and $i\starrightarrow j$ in $D$ and such that $\nu_t^i \rightarrow \nu_t^j$ in $D_{t'}$ if $i\tailedrightarrow j$ in $D$.
\end{definition}

\begin{definition}[Rolling, contemporaneous effects]
\label{def:rollingExt}
    Let $D_{t'}$ be a DAG on nodes $V^{t'} \cup W^{t'}$ such that $V^{t'} = \cup_{i\in V} \{ \nu_0^i, \ldots, \nu_{t'}^i \}$ and $W^{t'} = \{\nu_0^i, i\in W  \}$ for disjoint sets $V$ and $W$. The \emph{rolled version} of $D_{t'}$ is the tailed directed graph, $D$, on nodes $V \cup W$ such that for $i,j \in V \cup W$ we have that $i\starrightarrow j$ if there exists $s\leq t$ such that $\nu_s^i \rightarrow \nu_t^j$ in $D_{t'}$ and $i\tailedrightarrow j$ if and only if there exists $t$ such that $\nu_t^i \rightarrow \nu_t^j$ in $D_{t'}$.
\end{definition}

The causal DAGs that we consider are acyclic. However, if the contemporaneous edges are different for different time lags, then the tailed part of the rolled graphs do not necessarily form an acyclic graph. We let $\ant(B)$ denote the set of nodes $k$ such that there exists a directed path from $k$ to $j\in B$ consisting of tailed directed edges only, $k \tailedrightarrow \ldots \tailedrightarrow j$, and we let $\deta (B)$ denote the set of nodes $k$ such that there exists a directed path from $j \in B$ to $k$ consisting of tailed directed edges only. We let $\pat (B)$ denote the set of nodes $k$ such that $k\tailedrightarrow j$, $j\in B$. We use the convention that $\pat (B) \cap B = \ant(B) \cap B = \deta(B)\cap B  = \emptyset$. We extend $\delta$-separation to tailed directed graphs in the following way. If $D$ is a tailed directed graph, we construct a DG, $D^-$, by ignoring the distinction between directed edges and tailed directed edges, i.e.,

\begin{align*}
i \rightarrow j \text{ in } D^- \text{ if and only if } i\starrightarrow j
\text{ in } D.
\end{align*} 

\noindent We say that $B$ is $\delta$-separated from $A$ given $C$ in $D$ if $B$ is $\delta$-separated from $A$ given $C$ in $D^-$. The following example illustrates that we need to generalise Proposition \ref{prop:markov} when allowing contemporaneous effects.

\begin{example}
We consider graph $\mathbf{D}$ in Figure \ref{fig:roll} and denote it by $D$. This is an extended local independence graph with contemporaneous effects from $Q$ to $R$ and from $S$ to $Q$, and Graph $\mathbf{A}$ is the same graph but with directed edges instead of tailed directed edges. We see that the only path between $R$ and $S$ in $D^{\{R\}}$ is $S \starrightarrow Q \starrightarrow R$ and this is closed when conditioning on $Q$. This means that $R$ is $\delta$-separated from $S$ given $Q$ in graphs $\mathbf{A}$ and $\mathbf{D}$ when simply interpreting tailed directed edges as directed edges. In graph $\mathbf{B}$, this corresponds to the fact that there are no $d$-connecting paths between $\nu_2^R$ and $\{\nu_0^S,\nu_1^S\}$ given $\{\nu_0^Q,\nu_1^Q,\nu_0^R,\nu_1^R\}$. However, in the presence of contemporaneous effects this does not hold as $\nu_1^S \rightarrow \nu_2^Q \rightarrow \nu_2^R$ is $d$-connecting given $\{\nu_0^Q,\nu_1^Q,\nu_0^R,\nu_1^R\}$ in graph $\mathbf{E}$. This illustrates that if we simply interpret the tailed directed edges as directed edges, then Proposition \ref{prop:markov} does not hold.
\end{example}

In the mediation analysis, we will assume that the contemporaneous effects are from the outcome process to the other processes. We imagine that whenever a subject is examined we first obtain their survival indicator. Depending on survival, we may or may not obtain other measurements, e.g., blood pressure. This implicitly defines a partial ordering of the contemporaneous variables in which survival is first. Other contemporaneous edges may also be added. 

The next theorem gives a general approach to the graphical representation of local independence in time series with contemporaneous effects and it generalises Proposition \ref{prop:markov} which does not allow contemporaneous effects.  We let $\antv(B)$ denote the set of non-baseline nodes that are tailed ancestors of $B$, that is, $\antv(B) = \ant(B) \cap V$.

\begin{theorem}[$\delta$-separation Markov property with contemporaneous effects]
    Assume $(D,D_{t'})$ is proper and assume that the $d$-separation Markov property holds in $D_{t'}$. Let $A,C \subseteq V\cup W$ and $B\subseteq V$ be disjoint. If $A\cap \ant(B) = \emptyset$ and $A \not\rightarrow_\delta (\antv(B) \cap C) \cup B \mid C \setminus \antv(B)$, then $X^B$ is locally independent of $X^A$ given $X^C$ until time ${t'}$.
\label{thm:markovContemp}
\end{theorem}

\begin{figure}
\centering
        \begin{subfigure}[t]{0.45\textwidth}
                \centering
                \subcaption{}
                \includegraphics[width=.75\linewidth]{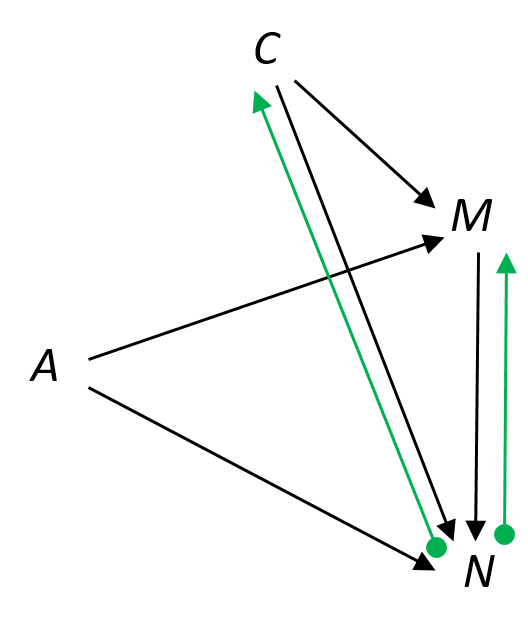}
                \label{graph of processes}
        \end{subfigure}\hfill
        \begin{subfigure}[t]{0.45\textwidth}
                \centering
                \subcaption{}
                \includegraphics[width=\linewidth]{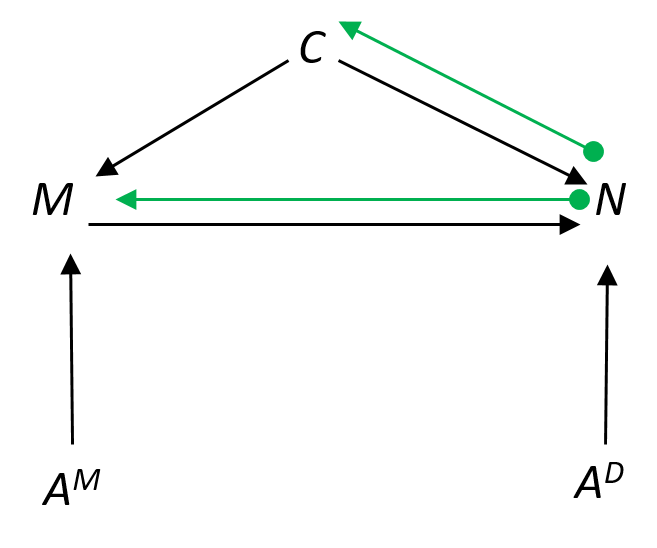}
                \label{local_independence_1}
        \end{subfigure}\hfill
        \caption{(\subref{graph of processes}) Extended local independence graph illustrating the relationship between treatment $A$, mediator process $M$, confounder process $C$ and outcome process $N$. The green arrows represent causal effects from the counting process $N$ as well as contemporaneous effects from the survival time. (\subref{local_independence_1}) Extended local independence graph in which Assumptions A1 and A2 hold (Proposition \ref{contempAssump}). $M$ denotes the mediator process and $N$ is a counting process (jumping to 1 when there is an event, at which time the process is stopped). $C$ denotes the observed confounders which may be either baseline variables or processes. The nodes $A$, $A^M$ and $A^D$ denote treatment and treatment components which are fixed at baseline.}
\end{figure}

\section{Causal assumptions and $\delta$-separation}
\label{singmed}

We assume that the causal model is represented by a DAG. In the previous section, we showed how a local independence graph can be constructed from this DAG and in this section we will argue that using $\delta$-separation one may deduce the validity of the mediation assumptions from this local independence graph. We now go back to the setting of the causal survival model described in Subsection \ref{basic model}.

\begin{proposition}
Assume a discrete-time causal model such that the global Markov property holds with respect to each unrolled graph, $D_t$, and such that the only contemporaneous effects are from $N$ to $M$ and from $N$ to $C$. We let $D$ denote the rolled graph such that $i\rightarrow j$ is in $D$ if it is in the rolled version of $D_t$ for any $t$ and such that $i\tailedrightarrow j$ is in $D$ if it is in the rolled version of $D_t$ for any $t$. We have that $A^D \not \rightarrow_\delta M \mid A^M, C, N$ in $D$ implies A1 and that $A^M \not \rightarrow_\delta C \mid A^D, M, N$  in $D$ implies A3. Moreover, $A^M \not \rightarrow_\delta N \mid A^D, C, M$ implies the independence in A2 for $t = t_{k+1}$, that is, a discrete version of A2.
\label{contempAssump}
\end{proposition}

The above proposition means that the assumptions A1-A3 in Section \ref{causmed} can be represented by $\delta$-separations in an extended local independence graph. As an example, consider the graph in Figure  \ref{local_independence_confounders_d}. It holds that $A^D  \not\rightarrow_{\delta} M \thinspace \mid  \thinspace A^M,C,N$. Using the above proposition, this means that $A^D \mathrel{\text{\scalebox{1.07}
{$\perp\mkern-10mu\perp$}}}  M_{k} \thinspace  \mid   \thinspace
\overline{C}_{k-1}, \overline{N}_{k}, A^M, \overline{M}_{k-1}$ corresponding to Assumption A1. A similar correspondence may be given for assumption A3. Assumption A2 does not formally correspond to a $\delta$-separation condition on the local independence graph as it uses a continuum of time points. It may be represented heuristically through the $\delta$-separation $A^M \not\rightarrow_\delta N \mid A^D, M, C$. In case of a discrete-time outcome process, this $\delta$-separation does in fact imply A2.

 We use the local independence graph representation for two reasons. First, the local independence graphs are simpler than their corresponding DAGs as they have fewer nodes and edges, and they therefore work better as communicative tools. One may also use a DAG representation with only two time lags to achieve something similar. Second, the local independence graphs can also be used in continuous-time models in which unrolled graphical representations do not correspond to the data-generating mechanism. Appendix \ref{ssec:continuousTimeExample} gives an example of mediation analysis using local independence graphs in a continuous-time multivariate stochastic process. 

\subsection{Examples}
\label{unmeasured confounders}

We now give some examples to illustrate the connection between local independence graphs and mediation analysis. As described in Subsection \ref{causmed}, we assume that the treatment, $A$, consists of two components and this means that our hypothetical experiment has two treatment variables, $A^D$ and $A^M$. Figure \ref{graph of processes} represents observable data while Figure \ref{local_independence_1} represents the hypothesised experiment in which treatment components $A^D$ and $A^M$ can take different values.

\begin{figure}
\centering
        \begin{subfigure}[b]{0.45\textwidth}
                \centering
                \includegraphics[width=\linewidth]{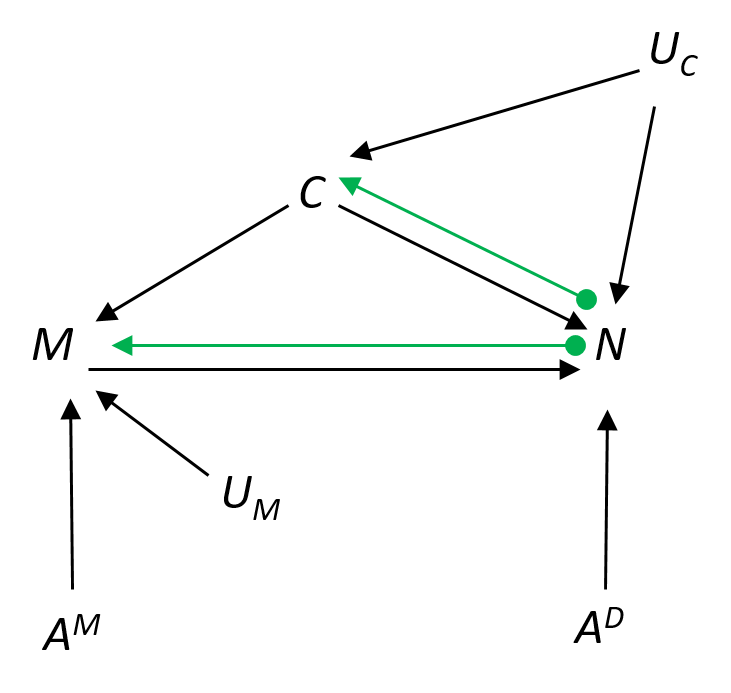}
                \caption{Extended local independence graph in which Assumptions A1 and A3 can be verified using Proposition \ref{contempAssump}. $M$ denotes the mediator process and $N$ the counting process. $C$ denotes the observed confounders which may be given at baseline or be a process over time. $U_M$ and $U_C$ denote unmeasured confounder processes. The nodes $A^M$ and $A^D$ denote treatment components fixed at baseline.}
            \label{local_independence_confounders}
        \end{subfigure}\hfill
        \begin{subfigure}[b]{0.45\textwidth}
                \centering
                \includegraphics[width=\linewidth]{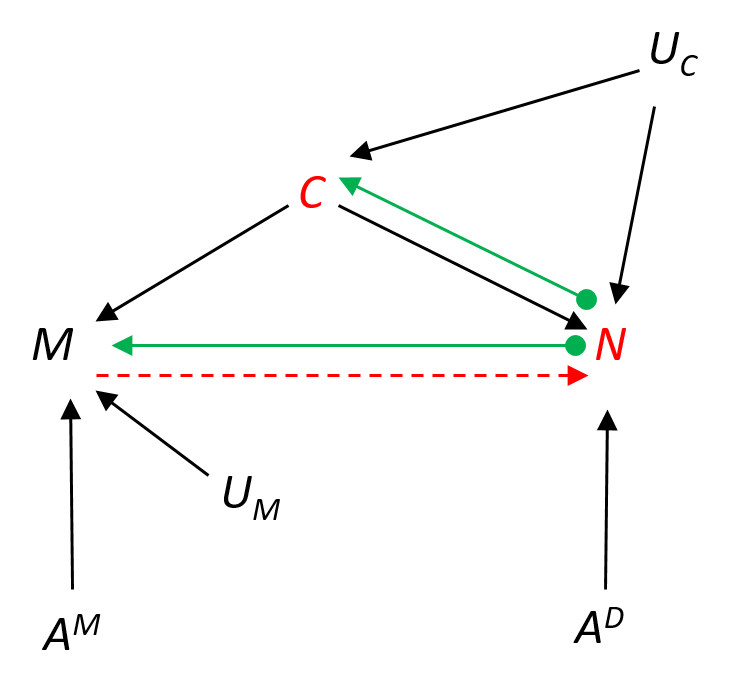}
                \caption{We apply the definition of  $\delta$-separation to the graph in Figure \ref{local_independence_confounders} to decide if $M$ is $\delta$-separated from $A^D$ given $\{A^M,C, N\}$. Notice that the directed edge from $M$ to $N$, here marked in red, is removed to construct the auxiliary graph $D^{\{M\}}$. }
        \label{local_independence_confounders_d}
        \end{subfigure}\hfill
        \begin{subfigure}[b]{0.45\textwidth}
                \centering
                \includegraphics[width=\linewidth]{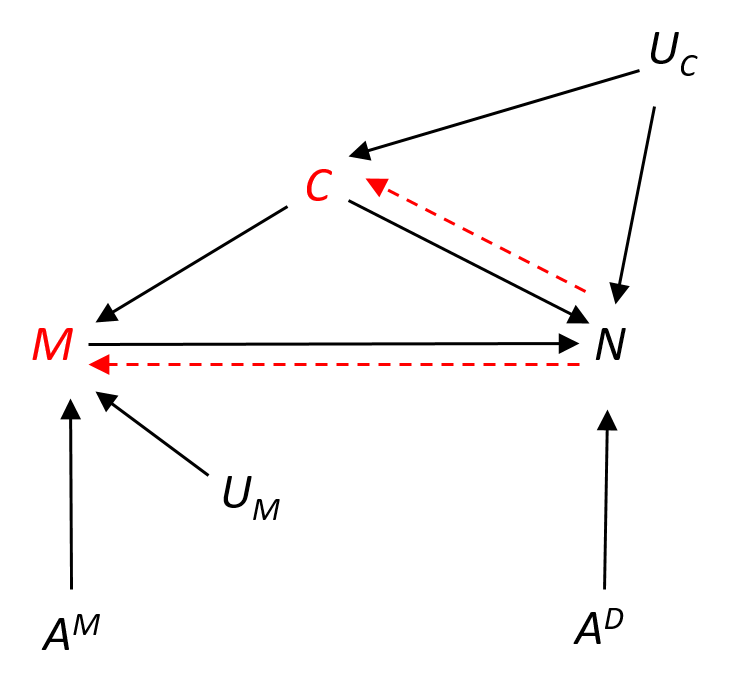}
                \caption{We apply the definition of  $\delta$-separation to the graph in Figure \ref{local_independence_confounders} to decide if $N$ is $\delta$-separated from $A^M$ given $\{A^D, C,M\}$. Notice that the directed edges out of $N$ in Figure \ref{local_independence_confounders}, here marked in red, are removed to construct $D^{\{N\}}$. The $\delta$-separation holds and this, informally, implies Assumption A2 using Proposition \ref{contempAssump}.}
                \label{local_independence_confounders_e}
        \end{subfigure}\hfill
        \begin{subfigure}[b]{0.45\textwidth}
                \centering
                \includegraphics[width=\linewidth]{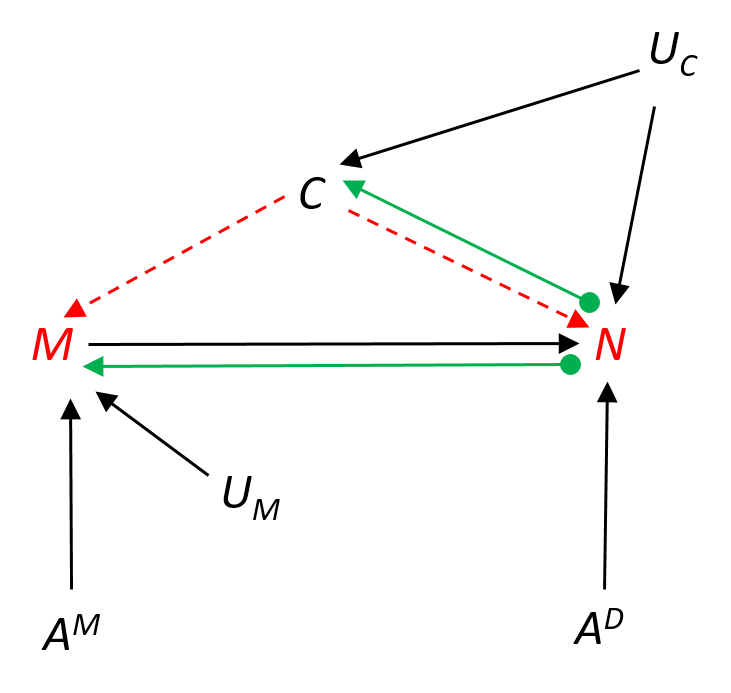}
                \caption{We apply the definition of  $\delta$-separation to the graph in Figure \ref{local_independence_confounders} to decide if $C$ is $\delta$-separated from $A^M$ given $\{A^D, M, N\}$. Notice that the directed edges out of $C$ in Figure \ref{local_independence_confounders}, here marked in red, are removed to construct $D^{\{C\}}$. One sees that the $\delta$-separation holds and this implies Assumption A3 using Proposition \ref{contempAssump}.}
                \label{local_independence_confounders_f}
        \end{subfigure}
        \caption{}\label{fig:li2}
\end{figure}

This approach can also be used in the presence of unmeasured confounder processes, $U^M$ and $U^C$, such as those in Figure \ref{local_independence_confounders}, see also \cite{vansteelandt2019mediation}.  Using Proposition \ref{contempAssump}, $\delta$-separation allows us to verify the assumptions by consulting the local independence graph.

\begin{example}
\label{gen}
We consider the graphs in Figure \ref{fig:li2} and assume that the graphs are rolled versions of the corresponding causal DAGs. These graphs all represent the hypothetical intervention in which $A^D$ and $A^M$ need not be equal. Note that in Figure \ref{local_independence_confounders} the event process $N$ may influence the mediator process $M$. This is relevant when $N$ does not only represent a single occurrence, but instead a number of events in a recurrent process, and this extends model represented in Figures \ref{graph of processes} and \ref{local_independence_1}.

Assumption A1 corresponds, loosely speaking, to the process $M$ being locally independent of $A^D$ given $A^M$, $C$ and $N$. We apply the $\delta$-separation criterion (Definition \ref{def:deltaSep}) to check whether this holds in the model represented by the graph, $D$, in Figure \ref{local_independence_confounders}. We construct the graph $D^{\{ M\}}$ by removing the directed edge from $M$ to $N$, marked in red in Figure \ref{local_independence_confounders_d}. We see that all remaining paths between $A^D$ and $M$ contain either $N$ or $C$ as a noncollider. As both these nodes are in the conditioning set, we see that $M$ is $\delta$-separated from $A^D$ given $\{A^M, C, N\}$. This implies that Assumption A1 holds using Proposition \ref{contempAssump}.

The $\delta$-separation rules are useful as it would be hard to judge intuitively which unmeasured confounders might invalidate the mediation analysis. As an example of this, consider the inclusion of an edge from $U_M$ to $N$ in Figure \ref{local_independence_confounders}. This would open a path between $A^D$ and $M$ in Figure \ref{local_independence_confounders_d} such that every collider is in $\an^\texttt{+}(\{A^M,C,N\})$ and no noncollider is in $\{A^M,C,N\}$. Therefore, $M$ is not $\delta$-separated from $A^D$ given $\{A^M,C,N\}$ and the validity of Assumption A1 does not follow.

Figures \ref{local_independence_confounders_e} and \ref{local_independence_confounders_f} show graphs $D^{\{N\}}$ and $D^{\{C\}}$ that can be used to check $\delta$-separations that imply the validity of Assumptions A2 and A3, respectively.
\end{example}


\begin{figure}[ptb]
\centering
\includegraphics[width=0.65\textwidth]{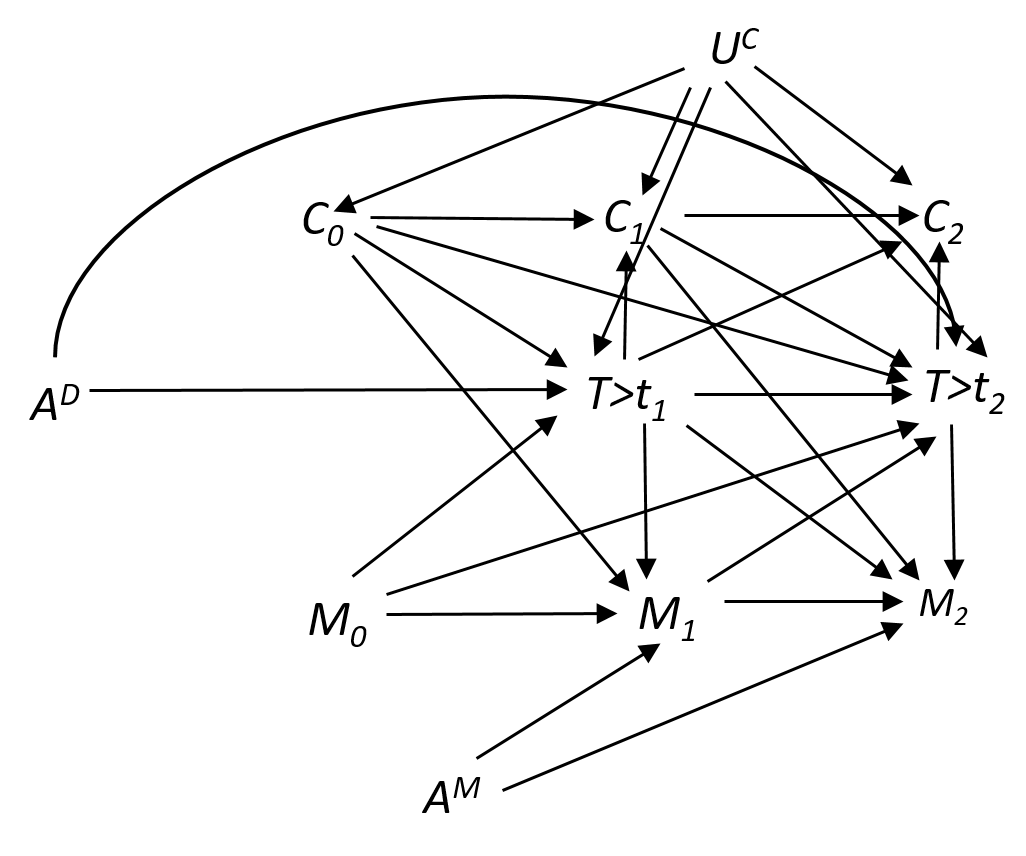} \caption{Causal directed acyclic graph for the first three time points.}%
\label{DAG_confounders}%
\end{figure}

Instead of the above approach, one may use $d$-separation in the underlying DAGs. However, these graphs will be more complex and the simplicity of the rolled graphs may make them more suitable when communicating with subject matter experts. An example DAG is given in Figure \ref{DAG_confounders}. Note that the unmeasured confounder, $U^C$, in Figure \ref{DAG_confounders} is a scalar variable, but it could just as well be a process.

\begin{example}
\label{partial}

We consider briefly a more complex setting as shown in Figure \ref{fig:li1}. The corresponding graphs showing $\delta$-separation are given in Figures \ref{local_independence_partial_a}, \ref{local_independence_partial_b} and \ref{local_independence_partial_c}. The graphs show the validity of Assumptions A1, A2 and A3. Hence, we have the same type of effects as in Example \ref{gen} and they are identified from observational data. There is also an apparent mediation through $C$, but this remains part of the direct effect in the present setup. Note that in this setup we need a model for how process $C$ depends on $A^D$.

There is a relationship to path-specific effects as discussed in Section 5.5.1 of \cite{vanderweele15}, see also \cite{shpitser2013counterfactual} and \cite{robins2022interventionist}. We believe a development of path-specific effects can be done within our framework, but further development must be made.

The example used here illustrates how a somewhat more complex mechanistic model can be considered. In practice one will have limited knowledge as to the correct mechanistic model, but one might try different alternatives and compare their effects. We shall not go into further details concerning this example.

\end{example}

\begin{figure}
\centering
        \begin{subfigure}[t]{0.45\textwidth}
                \centering
                \includegraphics[width=\linewidth]{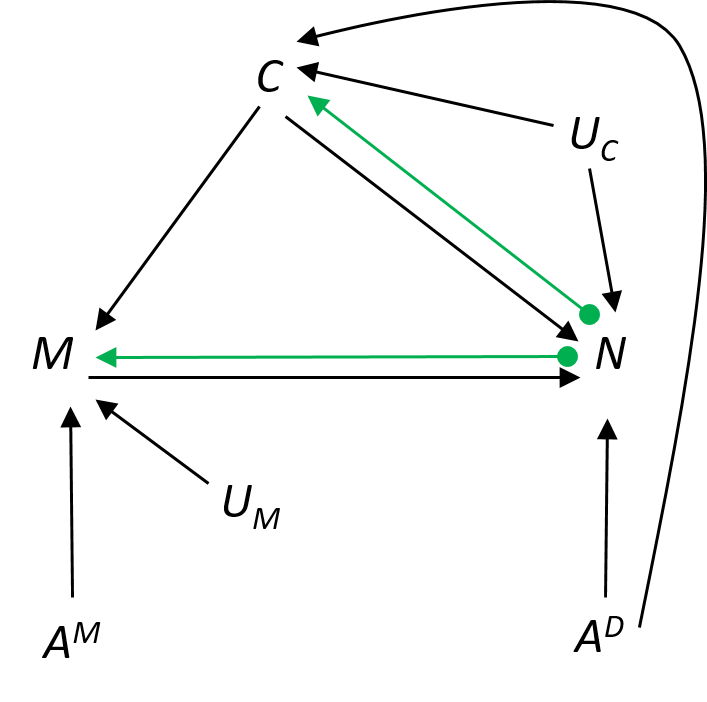}
                \caption{Extended local independence graph corresponding to time-continuous versions of  assumptions A1 and A2. $M$ denotes a mediator process and $N$ the counting process. $C$ also denotes a mediator process in this setting. $U_M$ and $U_C$ denote unmeasured confounder processes. The nodes $A^M$ and $A^D$ denote mediated and direct treatment components fixed at baseline. Note that the mediated effect is the one that goes through the process $M$ after an immediate effect from the treatment $A$. Note from the figure there will also be mediation through $C$ and further through $M$, but this shall not be included and will remain a part of the direct effect.}
            \label{local_independence_partial}
        \end{subfigure}\hfill
        \begin{subfigure}[t]{0.45\textwidth}
                \centering
                \includegraphics[width=\linewidth]{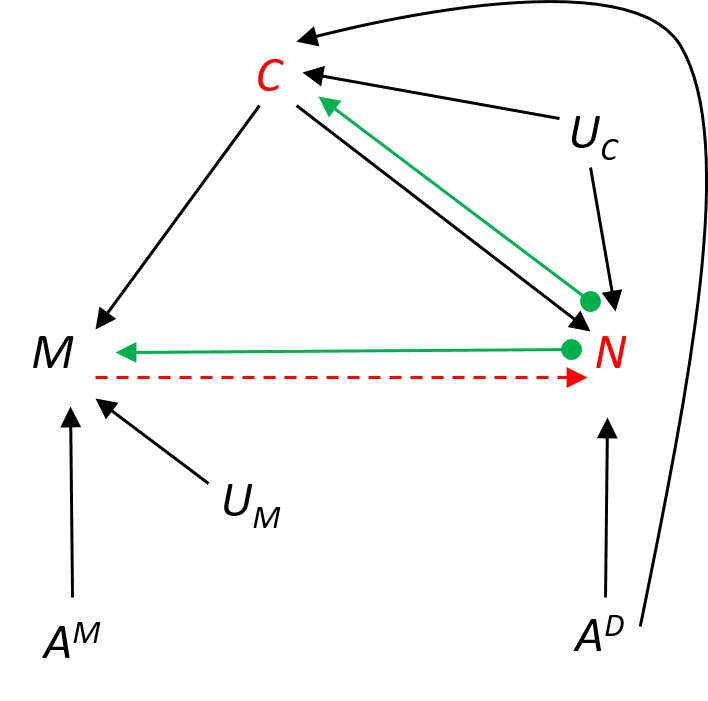}
                                \caption{Applying the directed path version of  $\delta$-separation to the graph in Figure \ref{local_independence_partial}, thus proving the local independence of $A^D$ on $M$ given $C$, $N$ and $A^M$. Notice that the directed node out of $M$ to $N$ in Figure \ref{local_independence_partial} (here marked in red) is removed as required by the $\delta$-separation rule. One sees that nodes $C$ and $N$ (marked in red) together block all paths from $A^D$ to $M$. This corresponds to assumption A1.}
            \label{local_independence_partial_a}
        \end{subfigure}\hfill
        \begin{subfigure}[t]{0.45\textwidth}
                \centering
                \includegraphics[width=\linewidth]{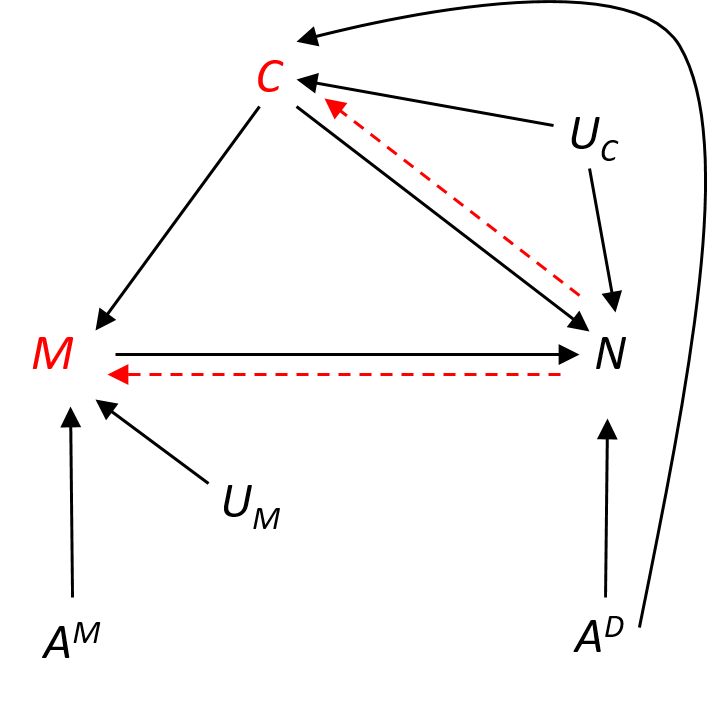}
                \caption{Applying the directed path version of  $\delta$-separation to the graph in Figure \ref{local_independence_partial}, thus proving the local independence of $A^M$ on $N$ given $C$, $M$ and $A^D$. Notice that the directed nodes out of $N$ to $M$ in Figure \ref{local_independence_partial} (here marked in red) are removed as required by the $\delta$-separation rule. One sees that nodes $C$ and $M$ (marked in red) together block all paths from $A^M$ to $N$. This corresponds to assumption A2.}
                \label{local_independence_partial_b}
        \end{subfigure}\hfill
        \begin{subfigure}[t]{0.45\textwidth}
                \centering
                \includegraphics[width=\linewidth]{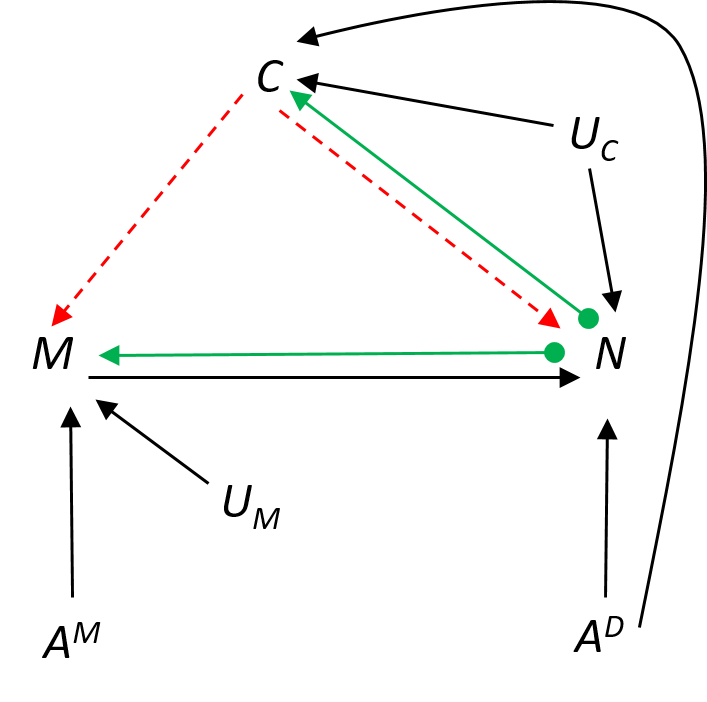}
                \caption{Applying the directed path version of  $\delta$-separation to the graph in Figure \ref{local_independence_partial}, thus proving the local independence of $A^M$ on $C$ given $N$, $M$ and $A^D$. Notice that the directed nodes out of $C$ in Figure \ref{local_independence_partial} (here marked in red) are removed as required by the $\delta$-separation rule. One sees that nodes $N$ and $M$ together block all paths from $A^M$ to $C$. This corresponds to assumption A3.}
                \label{local_independence_partial_c}
        \end{subfigure}
        \caption{Example graphs}\label{fig:li1}
\end{figure}

\section{A mediational g-formula}
\label{totpr}

In this section, we provide a mediational g-formula using Assumptions A1, A2 and A3. We shall express survival as an explicit function of the treatment, mediators and covariates, and then integrate out the mediators and covariates. We allow the presence of unmeasured confounders as in Section \ref{unmeasured confounders}.

Using the time ordering of the variables (Section \ref{basic model}, see also Figure \ref{DAG_confounders1}), we can write the joint distribution as a product of conditional probabilities. For $t_{k} < t \leq t_{k+1}$ we shall write this distribution conditionally on separate interventions on the treatment components expressed by $\mbox{do} (A^D=a, A^M=a^{\ast})$:

\begin{align}
& P(T>t,\overline{M}_{k}=\overline{m}_{k},\overline{C}_{k}=\overline{c}_{k} \thinspace\mid \thinspace \mbox{do} (A^D=a, A^M=a^{\ast})) \label{bigform}\\
& = P(T>t, T>t_1, T>t_2, \ldots,T>t_k,\overline{M}_{k}=\overline{m}_{k},\overline{C}_{k}=\overline{c}_{k} \thinspace\mid \thinspace \mbox{do} (A^D=a, A^M=a^{\ast})) \nonumber \\
& =P(T > t\mid T > t_{k},\overline{M}_{k}=\overline{m}_{k},\overline{C}_{k}=\overline{c}_{k}, \mbox{do} (A^D=a, A^M=a^{\ast})) \nonumber \\
& \times \prod_{i: t_i < t} \Big\{P(T > t_i \thinspace\mid \thinspace T > t_{i-1},\overline{M}_{i-1}=\overline{m}_{i-1},\overline{C}_{i-1}=\overline{c}_{i-1},
\mbox{do} (A^D=a, A^M=a^{\ast})) \nonumber \\
& \times P(M_{i}=m_{i}\thinspace\mid \thinspace T > t_i,\overline{M}_{i-1}=\overline{m}_{i-1},\overline{C}_{i-1}=\overline{c}_{i-1},
\mbox{do} (A^D=a, A^M=a^{\ast})) \nonumber \\
& \times P(C_{i}=c_{i}\thinspace\mid \thinspace T > t_i,\overline{M}_{i}=\overline{m}_{i},\overline{C}_{i-1}=\overline{c}_{i-1},
\mbox{do} (A^D=a, A^M=a^{\ast})) \nonumber \Big\}. 
\end{align}

We assume throughout that the conditional probabilities we need are well-defined, and, following \cite{lok2004estimating}, we use $-1$ as an `empty' index, that is, e.g., $P(M_{i}=m_{i}\mid T > t_i,\overline{M}_{i-1}=\overline{m}_{i-1},\overline{C}_{i-1}=\overline{c}_{i-1}, \mbox{do} (A^D=a, A^M=a^{\ast}))$ is to be read as $P(M_{0}=m_{0} \thinspace\mid \thinspace T > t_0, 
\mbox{do} (A^D=a, A^M=a^{\ast})) = P(M_{0}=m_{0} \thinspace\mid \thinspace 
\mbox{do} (A^D=a, A^M=a^{\ast}))  $ when $i=0$.

Following the approach of \cite{didelez2019defining} we shall show how the factors in the product in Equation (\ref{bigform}) can be simplified. We start with the factors concerning survival:

\begin{align*}
& P(T > t\mid T > t_{k},\overline{M}_{k}=\overline{m}_{k},\overline{C}_{k}=\overline{c}_{k}, \mbox{do} (A^D=a, A^M=a^{\ast})) \\
& = P(T > t\mid T > t_{k},\overline{M}_{k}=\overline{m}_{k},\overline{C}_{k}=\overline{c}_{k}, \mbox{do} (A^D=a, A^M=a)) \\
& = P(T > t\mid T > t_{k},\overline{M}_{k}=\overline{m}_{k},\overline{C}_{k}=\overline{c}_{k}, \mbox{do} (A=a)) 
\end{align*}

The first equality follows from Assumption A2 since $T>t$ is then d-separated from $A^M$. The last equality follows from the assumption of Property P1 of \cite{didelez2019defining}. Applying a similar argument to the factor concerning mediation, using Assumption A1, we get:

\begin{align*}
& P(M_{i}=m_{i}\thinspace\mid \thinspace T > t_i,\overline{M}_{i-1}=\overline{m}_{i-1},\overline{C}_{i-1}=\overline{c}_{i-1},
\mbox{do} (A^D=a, A^M=a^{\ast})) \\
& = P(M_{i}=m_{i}\thinspace\mid \thinspace T > t_i,\overline{M}_{i-1}=\overline{m}_{i-1},\overline{C}_{i-1}=\overline{c}_{i-1},
\mbox{do} (A=a^{\ast})) 
\end{align*}

Finally, we consider the factor concerning covariates in Equation (\ref{bigform}). Using Assumption A3, we find:

\begin{align}
&  P(C_{i}=c_{i}\thinspace\mid \thinspace T > t_i,\overline{M}_{i}=\overline{m}_{i},\overline{C}_{i-1}=\overline{c}_{i-1},
\mbox{do} (A^D=a, A^M=a^{\ast})) \label{separa} \\
& = P(C_{i}=c_{i}\thinspace\mid \thinspace T > t_i,\overline{M}_{i}=\overline{m}_{i},\overline{C}_{i-1}=\overline{c}_{i-1},
\mbox{do} (A=a))  \nonumber
\end{align}

Summing over $\overline{m}_{k}$ and $\overline{c}_{k}$ in Equation (\ref{bigform}), and using the results developed in this section, we get the following (for continuous mediators the sum would be substituted with integrals).

\begin{align}
& P(T>t \thinspace\mid \thinspace \mbox{do} (A^D=a, A^M=a^{\ast})) \label{bigform1}\\
& =  \sum_{\overline{m}_k,\overline{c}_k} \Big[  P(T>t,\overline{M}_{k}=\overline{m}_{k},\overline{C}_{k}=\overline{c}_{k} \thinspace\mid \thinspace \mbox{do} (A^D=a, A^M=a^{\ast})) \Big] \nonumber \\
& = \sum_{\overline{m}_k,\overline{c}_k} \Big[   P(T > t\mid T > t_{k},\overline{M}_{k}=\overline{m}_{k},\overline{C}_{k}=\overline{c}_{k}, \mbox{do} (A=a)) \nonumber \\
& \times \prod_{t_i \leq t} \Big\{P(T > t_i \thinspace\mid \thinspace T > t_{i-1},\overline{M}_{i-1}=\overline{m}_{i-1},\overline{C}_{i-1}=\overline{c}_{i-1},
\mbox{do} (A=a)) \nonumber \\
& \times P(M_{i}=m_{i}\thinspace\mid \thinspace T > t_i,\overline{M}_{i-1}=\overline{m}_{i-1},\overline{C}_{i-1}=\overline{c}_{i-1},
\mbox{do} (A=a^{\ast})) \nonumber \\
& \times  P(C_{i}=c_{i}\thinspace\mid \thinspace T > t_i,\overline{M}_{i}=\overline{m}_{i},\overline{C}_{i-1}=\overline{c}_{i-1},
\mbox{do} (A=a))  \Big\} \Big]  .
\end{align}

We define $r(t)=i$ when $t_{i}\leq t<t_{i+1}$ and rewrite the equation by introducing the hazard rate $\lambda(t\mid \overline{M}_{r(t)}=\overline{m}_{r(t)},\overline{C}_{r(t)}=\overline{c}_{r(t)}, A=a)$. Using Assumption A0 we can then write for $t_{i}\leq t<t_{i+1}$:

\begin{align*}
& P(T > t\mid T > t_{i},\overline{M}_{i}=\overline{m}_{i}, \overline{C}_{i}=\overline{c}_{i}, A=a) \\
& = \exp \left\{-\int_{t_{i}}^t \lambda(s\mid \overline{M}_{r(s)}=\overline{m}_{r(s)}, \overline{C}_{r(s)}=\overline{c}_{r(s)}, A=a) ds  \right\}
\end{align*}

Hence, the mediational g-formula can be written as follows (using that $A$ is randomly assigned):

\begin{equation}
\begin{split}
& P(T>t\mid \mbox{do} (A^D=a, A^M=a^{\ast})) \\
&  =  \sum_{\overline{m}_k,\overline{c}_k} \Big[ \exp \big\{-\int_0^t \lambda(s\mid \overline{M}_{r(s)}=\overline{m}_{r(s)}, \overline{C}_{r(s)}=\overline{c}_{r(s)}, A=a) ds  \big\}  \label{medg}\\ 
&\times  \prod_{t_i \leq t} \big\{  P(M_{i}=m_{i}\mid T > t_i,\overline{M}_{i-1}=\overline{m}_{i-1},\overline{C}_{i-1}=\overline{c}_{i-1}, A=a^{\ast}) \\
& \times  P(C_{i}=c_{i}\thinspace\mid \thinspace T > t_i,\overline{M}_{i}=\overline{m}_{i},\overline{C}_{i-1}=\overline{c}_{i-1},
A=a)  \big\} \Big].    
\end{split}
\end{equation}

The present mediational g-formula is stated in terms of hazard functions. These can be estimated under censoring and we assume non-informative (independent) censoring. Formula (\ref{medg}) can be used with different hazard rate models.

This mediational g-formula is similar to the result in Section 3.2 of \cite{vansteelandt2019mediation}, but with somewhat different assumptions. Note also that by setting $a=a^{\ast}$ in the mediational g-formula (\ref{medg}) we get the standard g-computation formula \citep{robins1986new,lok2004estimating}:

\begin{equation}
\begin{split}
& P(T>t\mid \mbox{do} (A=a)) \\
&  =  \sum_{\overline{m}_k,\overline{c}_k} \Big[ \exp \big\{-\int_0^t \lambda(s\mid \overline{M}_{r(s)}=\overline{m}_{r(s)}, \overline{C}_{r(s)}=\overline{c}_{r(s)}, A=a) ds  \big\}  \label{gcomp}\\ 
&\times  \prod_{t_i \leq t} \big\{  P(M_{i}=m_{i}\mid T > t_i,\overline{M}_{i-1}=\overline{m}_{i-1},\overline{C}_{i-1}=\overline{c}_{i-1},A=a) \\
& \times  P(C_{i}=c_{i}\thinspace\mid \thinspace T > t_i,\overline{M}_{i}=\overline{m}_{i},\overline{C}_{i-1}=\overline{c}_{i-1},
A=a) \big\} \Big].    
\end{split}
\end{equation}

\section{Estimation of direct and indirect effects}

Estimation of the probabilities in the mediational g-formula can be carried out along the lines indicated by \cite{vansteelandt2019mediation} or \cite{lin2017mediation}. An alternative is an additive hazard model for mediation analysis as presented by \cite{aalen2020time}.

In particular the approach of \cite{vansteelandt2019mediation} is very general and could be recommended in many cases. However, as a useful and simple alternative we will present a structural cumulative survival model. We will derive explicit estimates in the case when there is no unmeasured confounder $U^C$.

\subsection{A semiparametric additive hazards model}
\label{semipar}

 We shall apply the semiparametric additive hazards model, see \cite{martinussen2017instrumental,vansteelandt2018survivor,dukes2019doubly}. This is a much more flexible model than the ordinary additive hazards model. A general formulation of the model as a structural cumulative survival model (SNCSTM) is given in \cite{seaman2020adjusting} and \cite{seaman2021using}. A recent application is given by \cite{ying2021new}. The model allows a very simple derivation of direct and indirect effects.
 
 In our setting the hazard rate is put on the following form:

\begin{equation}
  \lambda(t\mid \overline{M}_{r(t)}=\overline{m}_{r(t)},\overline{C}_{r(t)}=\overline{c}_{r(t)}, A=a) 
  = \rho_ta + \lambda_1(t\mid \overline{m}_{r(t)},\overline{c}_{r(t)})
  \label{haz1}
\end{equation}
for some parameter function $\rho_t$ and a function $\lambda_1$ of the mediator processes and covariate processes, seen as time-dependent covariates. Note that the indirect part expressed through the $\lambda_1$-function could have any functional form, e.g. that of a Cox model. The hazard rate could depend on the past mediator values in many different ways; some examples are given next:

\begin{itemize}
    \item The hazard rate depends only on the last observed mediator value.
    \item The hazard rate depends on the mean of all previously measured values of mediators and covariates.
    \item The hazard rate depends on some weighted average of the previously measured values of the mediator, e.g., giving higher weights to more recent values.
    \item The hazard rate depends on a two-dimensional summary statistic: One part is a weighted mean of early observed mediator values, while the other is a weighted mean of the remaining values.
\end{itemize}

We introduce the hazard rate into Equation (\ref{haz1}) in the mediational g-formula given in Equation (\ref{medg}). Assume $t_{k}\leq t<t_{k+1}$, then:

\begin{equation*}
\begin{split}
& P(T>t\mid\mbox{do} (A^D=a, A^M=a^{\ast})) \\
&  =  \sum_{\overline{m}_k,\overline{c}_k} \Big[  \exp \big\{-\int_0^t (\rho_{s}a + \lambda_1(s\mid\overline{m}_{r(s)},\overline{c}_{r(s)}))ds \big\} \\ 
&\times  \prod_{t_i \leq t} \big\{  P(M_{i}=m_{i}\mid T > t_i,\overline{M}_{i-1}=\overline{m}_{i-1},\overline{C}_{i-1}=\overline{c}_{i-1},\mbox{do} (A=a^{\ast})) \\
& \times  P(C_{i}=c_{i}\thinspace\mid \thinspace T > t_i,\overline{M}_{i}=\overline{m}_{i},\overline{C}_{i-1}=\overline{c}_{i-1},
\mbox{do} (A=a)) \big\}\Big]
\end{split}
\end{equation*}

\noindent We now assume that we are in the setting of Example \ref{gen}, see also Figure \ref{fig:li2}.

\subsection{A simple analysis when there is no unmeasured confounder $U^C$}

We assume that Figure \ref{local_independence_confounders} represents the causal structure, but with no unmeasured confounder $U^C$. In this case, the above  mediational g-formula can be written as:

\begin{equation*}
\begin{split}
& P(T>t\mid \mbox{do} (A^D=a, A^M=a^{\ast})) \\
&  = \exp \big\{-\int_0^t \rho_{s}(a-a^{\ast})ds\big\} \\
&\times \sum_{\overline{m}_k,\overline{c}_k} \Big[  \exp \big\{-\int_0^t (\rho_{s}a^{\ast} +
\lambda_1(s\mid \overline{m}_{r(s)},\overline{c}_{r(s})) ds\\ 
&\times  \prod_{t_i \leq t} \big\{  P(M_{i}=m_{i}\mid T > t_i,\overline{M}_{i-1}=\overline{m}_{i-1},\overline{C}_{i-1}=\overline{c}_{i-1},\mbox{do} (A=a^{\ast})) \\
& \times P(C_{i}=c_{i}\thinspace\mid \thinspace T > t_i,\overline{M}_{i}=\overline{m}_{i},\overline{C}_{i-1}=\overline{c}_{i-1},
\mbox{do} (A=a^\ast))\Big]\\
&  = \exp \big\{-\int_0^t \rho_{s}(a-a^{\ast})ds\big\} \times P(T>t\mid \mbox{do} (A=a^{\ast}))
\end{split}
\end{equation*}

\noindent The first equation follows from taking an exponential factor containing $a-a^{\ast}$ out of the large sum and applying Rule 3 of do-calculus \citep{pearl09} to the underlying causal graph (note that the $d$-separations needed for this rule hold in the original graph and therefore also in the graph used in Rule 3). The last step follows from Property P1 as the sum then equals $P(T>t\mid \mbox{do} (A=a^{\ast})$ since only $a^{\ast}$ is in the sum, and not $a$.

Using the g-computation formula in Equation (\ref{gcomp}) we can write:

\begin{align}
& Q(t;a,a^{\ast})=P(T>t\mid \mbox{do} (A^D=a, A^M=a^{\ast})) \nonumber \\
& = \exp \big\{-\int_0^t \rho_{s}(a-a^{\ast})ds\big\} \times P(T>t\mid \mbox{do} (A=a^{\ast}))
\label{bigmed}
\end{align}

We define direct and indirect effects in terms of relative survival functions and denote these as $SIE$ and $SDE$ respectively:

\begin{align*}
SIE(t)  &  = Q(t;a,a)/Q(t;a,a^{\ast})=\exp  \big\{ (a-a^{\ast})\int_0^t\rho_{s}ds\big\} \frac{ P(T>t\mid \mbox{do} (A=a))}{ P(T>t\mid \mbox{do} (A=a^{\ast}))}   \\
SDE(t)  &  = Q(t;a,a^{\ast})/Q(t;a^{\ast},a^{\ast})= \exp \big\{(a^{\ast}-a) \int_0^t\rho_{s}ds \big\} \end{align*}
The total effect is then $P(T>t\mid A=a) / P(T>t\mid A=a^{\ast})$. Note that Assumption A0 implies that $P(T>t\mid A=a)=P(T>t\mid \mbox{do}(A=a))$.

\subsection{An estimation approach based on the Cox model}
\label{simple1}

We shall consider the following setting: We assume the model given in Figure \ref{fig:li2}. and we also assume that there is no unmeasured confounder $U^C$, see Section \ref{semipar}. When the treatment $A$ is randomised, the survival functions $P(T>t\mid \mbox{do} (A=a))$ and $P(T>t\mid \mbox{do} (A=a^{\ast}))$ can be estimated by Kaplan-Meier survival functions. What remains is then to estimate the parameter function $\rho_{t}$. If $\lambda_1(t\mid \overline{m}_{r(t)},\overline{c}_{r(t)})$ is on an additive form in the covariates, then the estimation can be made with an additive hazards model. This is a straightforward approach.

However, other functional forms for $\lambda_1(t\mid \overline{m}_{r(t)},\overline{c}_{r(t)})$ may also be used, e.g., a Cox model in which case the model in (\ref{haz1}) corresponds to the proportional excess hazards model of \cite{martinussen2007}. We will illustrate this approach here (details are given in Appendix \ref{cox-est} in the supplementary materials). We assume that the two possible values of the exposure variable $A$ are $a=0$ and $a=1$. Assume that $a=0$ indicates the group with the smallest hazard. For this group of individuals the hazard function $\lambda_1(t\mid \overline{m}_{r(t)},\overline{c}_{r(t)})$ can then be directly estimated from the Cox model.

More specifically, we define the Cox model as follows:

\begin{equation}
\label{aalen:cox}
\lambda_1(t\mid \overline{m}_{r(t)},\overline{c}_{r(t)}) = \psi(t) \exp{({\gamma}\thinspace \mathbf{z}_t)}
\end{equation}
Here $\mathbf{z}_t$ is defined as a time-dependent (vector) quantity based on the mediators and covariates. For example, there could be fixed covariates defined at time zero. In a simple case there could be one mediator defined, e.g., as the average over relevant past mediator values at a given time. The estimated (vector) parameter is the quantity $\widehat{\boldsymbol{\gamma}}$.

Let $N_0(t)$ and $N_1(t)$ be the counting processes for groups $a=0$ and $a=1$, respectively. The Breslow estimator of the underlying cumulative hazard (integral of $\psi(t)$) is given as follows:

\begin{equation}
\widehat{\Lambda}_0(t)= \int_0^t \frac{ dN_0(u)} 
{\sum_{l=1}^n (Y_l(u,a=0)\exp{(\widehat{\boldsymbol{\gamma}}
\thinspace \mathbf{z}_u(l))}}
\end{equation}
where $Y_l(u,a=0)$ is 1 when $a=0$ and the individual is still in the risk set, and zero otherwise. The quantity $\mathbf{z}_t(l)$ is the time-dependent covariate in the Cox model for individual $l$ at time $t$.

In the next step we use the individuals from the group with $a=1$ with $Y_l(u,a=1)$ being defined for this group in analogy with above. For each individual we insert the respective values of $\overline{m}_{r(t)}$ and $\overline{c}_{r(t)}$ into the Cox model estimated above for $\lambda_1(t\mid \overline{m}_{r(t)},\overline{c}_{r(t)})$, see formula (\ref{aalen:cox}). The next formula is derived in Appendix \ref{cox-est}:

\begin{equation}
\label{aalen:pripp}
 \int_0^t \frac{ \sum_{l=1}^n (Y_l(u,a=1)\exp{(\widehat{\boldsymbol{\gamma}}
\thinspace \mathbf{z}_u(l)})} 
{\sum_{l=1}^n (Y_l(u,a=1)
\sum_{l=1}^n (Y_l(u,a=0)\exp{(\widehat{\boldsymbol{\gamma}}
\thinspace \mathbf{z}_u(l)})} dN_0(u)
\end{equation}
where $Y_l(u,a=1)$ is 1 when $a=1$ and the individual is still in the risk set, and zero otherwise. 
Finally, we subtract the estimator in Equation (\ref{aalen:pripp}) from the Nelson-Aalen estimator calculated for the group $a=1$. The result is an estimate of the cumulative function $\int_0^t \rho_s ds$. For more details, see Appendix \ref{cox-est}.

\begin{equation}
\begin{split}
& \widehat{R(t)} = \int_0^t \frac{dN_1(u)}{\sum_{l=1}^n (Y_l(u,a=1)) }  \\
&  - \int_0^t \frac{ \sum_{l=1}^n (Y_l(u,a=1)\exp{(\widehat{\boldsymbol{\gamma}}
\thinspace \mathbf{z}_u(l)})} 
{\sum_{l=1}^n (Y_l(u,a=1)
\sum_{l=1}^n (Y_l(u,a=0)\exp{(\widehat{\boldsymbol{\gamma}}
\thinspace \mathbf{z}_u(l)})} dN_0(u)   
\end{split}
\end{equation}

\section{Data application}
\label{application}
To illustrate the suggested estimation procedure we use the liver cirrhosis dataset which is freely available as part of the timereg package in R \citep{martinussen2007}. 
The liver cirrhosis data stem from a randomised trial conducted in the Copenhagen area between 1962 and 1969, including patients with histologically verified liver cirrhosis who were randomised to either receive prednisone or placebo at baseline. They were scheduled for regular visits after 3, 6 and 12 months and yearly thereafter and followed until death, end of study on October 1st, 1974 or loss to follow-up. For detailed description of the trial see for example \cite{ schlichting1983csl}. 
Several versions of these data have previously been used to illustrate various methods for modelling time-to-event data with time-dependent covariates.  \cite{ andersen1993statistical} use it at several occasions in their book, especially to illustrate multistate models. More recently, the data were used in  \cite{fosen06b} to exemplify dynamic path analysis using the additive hazards model. In the following, we address one of the questions also addressed in \cite{fosen06b}, namely, whether and to which extent the effect of prednisone on time to death is mediated through prothrombin.

In the theory outlined in the previous sections, the separation of intervention effects into different components ($A^M$ and $A^D$) that affect the outcome via different pathways is key. 
Hence, in any given analysis the underlying biological mechanism should be considered carefully to gauge whether this assumption is sensible. 

A primary trigger for end-stage liver disease (cirrhosis), from which patients in our sample suffer, is excessive systemic inflammation. Glucocorticoids, such as prednisone, have the ability to promptly suppress excessive inflammatory reactions and immune response. Hence, glucocorticoids have been used in liver cirrhosis treatment for decades despite remaining uncertainties regarding all molecular mechanisms in action \citep{xue2019gluco}. 

Prothrombin is a protein synthesised by the liver and is among several other substances known as a clotting (coagulation) factor. In previous studies prothrombin has been considered as a measure of liver function \citep{andersen1991}.

Considering recent reviews on the effect of glucocorticoids on coagulation factors  \citep{ vanzaane2010gluco} and on coagulation abnormalities in cirrhotic patients \citep{ mucino2013caogulation}, it is not unreasonable to investigate a potential mediating effect through prothrombin. 
However, given the aforementioned uncertainties regarding the molecular mechanisms activated by glucocorticoid treatment, some uncertainty remains whether the different pathways through which treatment works could be activated by distinct treatment components ($A^M$ and $A^D$).

\subsection{Analysis of the effect of prednisone on death through prothrombin}

For our illustration, we used information on 446 patients, 226 in the active treatment group and 220 in the placebo group, among whom 270 deaths (n=131 active group, n=139 placebo group) were observed. Following \citet{fosen06b}, we transformed the repeatedly measured prothrombin values by setting values of 70 or above to zero and subtracting 70 if the original values were below 70. As prothrombin values above 70 are considered ‘normal’, the transformed prothrombin values are zero for prothrombin values considered to be normal and negative for lower-than-normal values. 

We evaluated direct and indirect effects considering (i) the most recent value of prothrombin as well as (ii) an average of all previously measured values of prothrombin considering age and sex as confounding variables. Figure \ref{local_independence_1} represents the assumptions of this mediation analysis graphically when $A$ is treatment, $M$ is prothrombin value (defined as in (i) or (ii)), $N$ is outcome and $C$ is age and sex. The estimated direct, indirect and total effects on a relative survival scale using the estimation procedure outlined in Section \ref{simple1} are shown in Figure \ref{rel_surv}. 
The black lines display effect estimates when using the most recent value of prothrombin as mediating variable. The gray lines indicate the corresponding confidence intervals based on 1000 bootstrap samples. The blue lines display the effect estimates when considering the average of all previously measures of prothrombin as the mediating variable.

\begin{figure}
\centering
\makebox[\linewidth]{\includegraphics[scale=0.5]{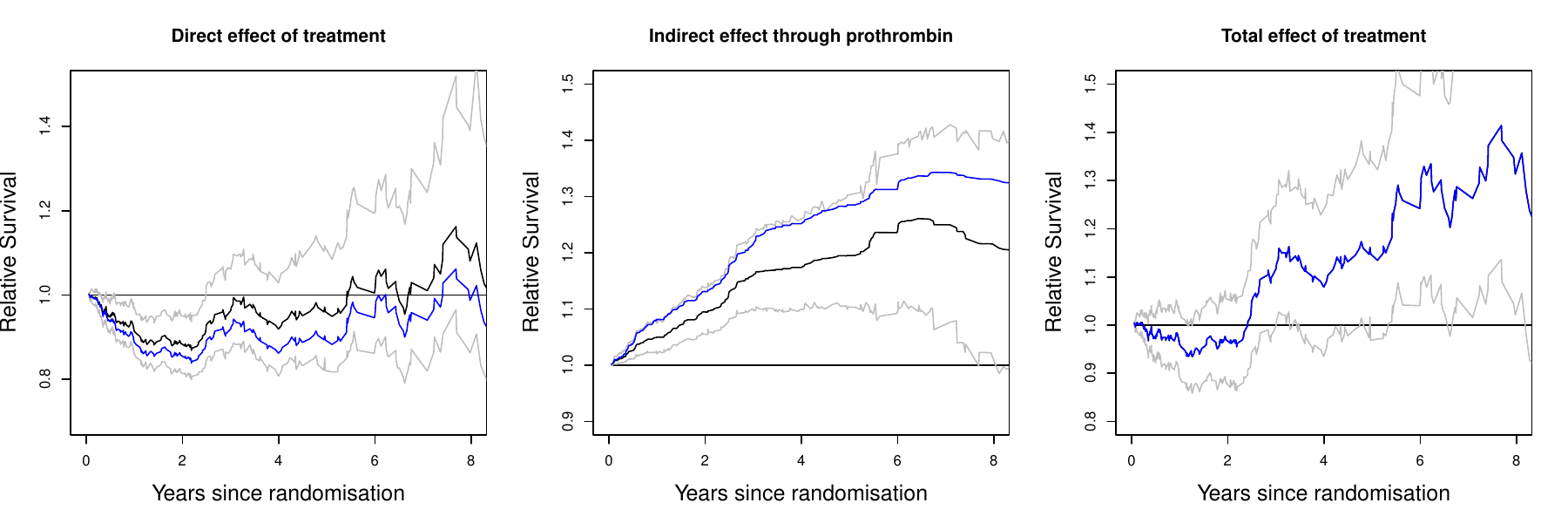}}

  \caption{Estimated direct, indirect and total effect on the relative survival scale comparing the effect of prednisone vs. placebo on mortality considering prothrombin as a mediator. The black lines show the effect estimates considering the most recent prothrombin value, corresponding 95\% confidence intervals are plotted in gray. The blue lines indicate the effects using the average of all previous prothrombin values as the mediator.}
   \label{rel_surv}
	\end{figure}

While the total effect of prednisone seems fairly constant over the first 2.5 years, the slope indicates a beneficial effect thereafter. Based on the obtained direct and indirect effects it seems that this beneficial effect is mainly mediated through prothrombin. As indicated by the blue lines, it appears that incorporating information on the whole mediator history explains more of the total effect than simply using the most recent value (black lines).

\citet{fosen06b} also considered the most recent prothrombin value as a mediator and report comparable behaviors of the total, direct and indirect effects on the cumulative hazards scale when using the additive hazards model. However, they had more information on clinical and non-clinical characteristics available to facilitate more detailed confounding adjustment and to restrict their analysis to 386 patients without ascites. Hence, our results are not directly comparable and should be interpreted cautiously from a medical perspective given the limited number of considered confounders. 

Taken together, the presented analysis mainly aims to illustrate the potential of the suggested estimation approach for this type of data structure and research questions. Corresponding R code is available on GitHub.

\section{Discussion}
We have developed a general approach for mediation analysis in survival models following the ideas of \cite{didelez2019defining}. Clearly, the treatment separation approach is somewhat speculative, which can in fact be said about almost all causal models. However, it does relate to biological hypotheses that could in principle be studied in more detail, as opposed to the nested counterfactuals which may not exist in any meaningful sense \citep{aalen12armitage}.  Undoubtedly, causal thinking is useful. It tells us which measures could have validity in a causal setting. Mediation analysis, in particular, is complex and a mathematical framework is necessary to understand the ideas.

Assumptions like A1, A2 and A3 in effect make statements about biological reality, so one has to be careful. However, one can also assume a pragmatic point of view: Mathematically, the assumptions represent orthogonality statements about the direct and mediated effects. So, one can say that the approach yields an orthogonalisation which may be a useful analysis, even if the biological reality is more complex. If one has more detailed knowledge about the biology one could also use more complex models as shown in Example \ref{partial}. The various models could also be used in sensitivity analyses. These are very important for evaluating the actual validity of causal models. They are not included here, but the basic mathematical approach of sensitivity analyses could follow the ideas presented here.

An important contribution of the present paper is our development and use of extended local
independence graphs as representations of the assumptions that facilitate the mediation analysis. These graphs are compact representations of the causal structure from which sufficient conditions for the mediation analysis may be read off. This also extends to time-continuous models as illustrated by an example in the appendix and future work may extend this framework to other classes of continuous-time models.

\section{Software}
The R code corresponding to the basic data example in Section \ref{application} is available on GitHub: \url{https://github.com/susistroh/time_depenent_mediator_mogensen_aalen_strohmaier}

\section{Supplementary materials}
Appendices \ref{app:granger}-\ref{ssec:continuousTimeExample} contain further 
results as well as proofs that were omitted in the main manuscript.

\section{Funding}

This work was supported by Independent Research Fund Denmark [DFF-International Postdoctoral Grant 0164-00023B to SWM]. 

\section{Acknowledgements}

The authors thank Vanessa Didelez and Mats Stensrud for their helpful comments on an earlier version of this manuscript. SWM is a member of the ELLIIT Strategic Research Area at Lund
University.

\appendix

\section{Granger Causality}
\label{app:granger}

In this section, we provide an alternative definition of Granger causality. The 
two definitions are equivalent and are essentially different parameterizations 
of the set of triples $(A,B,C)$ for which we can ask if $A$ is Granger 
non-causal for $B$ given $C$.

\begin{definition}[Granger non-causality, alternative]
	Let $X^V$ be a multivariate time series with an index set $V$. Let $A$ and 
	$B$ be disjoint subsets of $S \subseteq V$. We say that $X^A$ is Granger 
	non-causal for $X^B$ \emph{relative to} $X^{S}$ if for all $t\geq 1$ we 
	have that $\overline{X}_{t-1}^A$ and $X_t^B$ are conditionally independent 
	given $\overline{X}_{t-1}^{S\setminus A}$.
\end{definition}

\section{Estimating model in Section \ref{simple1}}
\label{cox-est}

Let $N_0(t)$ and $N_1(t)$ be the counting processes for groups $a=0$ and $a=1$, 
respectively. The counting process representation of $N_0(t)$ is given as 
follows:

\begin{equation}
\label{aalen:km}
dN_0(t) = \psi(t)dt \sum_{l=1}^n (Y_l(t,a=0)\exp{(\boldsymbol{\gamma} 
	\thinspace \mathbf{z}_t(l))} + dM_0(t)
\end{equation}
where $Y_l(t,a=0)$ is 1 when $a=0$ and individual $l$ is still in the risk set 
at time $t$, and zero otherwise. The history is defined as all information 
available up to a given time, and $M_0(t)$ is a martingale with respect to this 
history. The quantity $\mathbf{z}_t(l)$ is the time-dependent covariate in the 
Cox model for individual $l$ at time $t$.

In the next step we use the individuals from the group with $a=1$. For each 
individual we insert the respective values of $\overline{m}_{r(t)}$ and 
$\overline{c}_{r(t)}$ into the Cox model estimated for $\lambda_1(t\mid 
\overline{m}_{r(t)},\overline{c}_{r(t)})$, with time-dependent covariate values 
denoted $\mathbf{z}_t(l)$. The martingale representation of $N_1(t)$ is given 
as follows

\begin{equation}
\label{aalen:dill}
dN_1(t) = \rho_t \sum_{l=1}^n (Y_l(t,a=1)) + \psi(t) dt \sum_{l=1}^n 
(Y_l(t,a=1)\exp{(\boldsymbol{\gamma}
	\thinspace \mathbf{z}_t(l))} + dM_1(t)
\end{equation}
where $Y_l(t,a=1)$ is 1 when $a=1$ and the individual is still in the risk set, 
and zero otherwise.

The quantity $\psi(t)dt$ can be estimated from Equation (\ref{aalen:km}) as 
follows:

\begin{equation}
\psi(t)dt \approx \frac{dN_0(t)}{\sum_{l=1}^n 
(Y_l(t,a=0)\exp{(\boldsymbol{\gamma} 
		\thinspace \mathbf{z}_t(l))}}
\end{equation}
This is the Breslow estimator. Substituting this into Equation 
(\ref{aalen:dill}) yields:

\begin{equation}
\begin{split}
& \frac{dN_1(t)}{\sum_{l=1}^n (Y_l(t,a=1)) } = \rho_t  \label{kompli} \\
&  +  \frac{ \sum_{l=1}^n (Y_l(t,a=1)\exp{(\boldsymbol{\gamma}
		\thinspace \mathbf{z}_t(l)})} 
{\sum_{l=1}^n (Y_l(t,a=1)
	\sum_{l=1}^n (Y_l(t,a=0)\exp{((\boldsymbol{\gamma}
		\thinspace \mathbf{z}_t(l)})} dN_0(t)   
\end{split}
\end{equation}
When integrated over $t$ the left hand side of equation (\ref{kompli}) yields 
the steps of the Nelson-Aalen estimator calculated for the group $a=1$. 
Integrating the second line of equation (\ref{kompli}) yields a stochastic 
integral for group $a=0$. When subtracting this from the Nelson-Aalen estimator 
an estimate is derived for the cumulative function $R(t)=\int_0^t \rho_s ds$.

Above we have assumed that $\boldsymbol{\gamma}$ is a fixed quantity. In 
practice, it is estimated from the Cox model for the group $a=0$. The estimated 
quantity is termed $\widehat{\boldsymbol{\gamma}}$. Integrated into the 
estimator it yields

\begin{equation}
\begin{split}
& \widehat{R(t)} = \int_0^t \frac{dN_1(u)}{\sum_{l=1}^n (Y_l(u,a=1)) }  \\
&  - \int_0^t \frac{ \sum_{l=1}^n (Y_l(u,a=1)\exp{(\widehat{\boldsymbol{\gamma}}
		\thinspace \mathbf{z}_u(l)})} 
{\sum_{l=1}^n (Y_l(u,a=1)
	\sum_{l=1}^n (Y_l(u,a=0)\exp{(\widehat{\boldsymbol{\gamma}}
		\thinspace \mathbf{z}_u(l)})} dN_0(u).   
\end{split}
\end{equation}

\section{Proofs and additional results}
\label{app:proofs}

\begin{lemma}
	Let $(D,D_{t'})$ be a proper pair and let $A,C \subseteq V\cup W$ and 
	$B\subseteq V$ be disjoint. Let $E \subseteq V$ be a tail ancestral set 
	such that $\deta(B) \cap E= \emptyset$, $E\subseteq C$ and 
	$\pat(B)\subseteq E$. If $B$ is $\delta$-separated from A given C in $D$, 
	then $\nu_t^B$ is $d$-separated from $\Bar{\nu}_{t-1}^A$ given 
	$\bar{\nu}_{t-1}^{B\cup C}\cup \nu_t^E$ in $D_{t'}$, $0<t\leq t'$.
	\label{lem:extDeltasepToDsep}
\end{lemma}

The lemma is similar to a result in \cite{mogensen2020markov}, however, it 
generalises the result of that paper by including baseline variables and 
contemporaneous effects. It is also possible to state conditions such that 
$d$-separation implies $\delta$-separation, see, e.g., the supplementary 
material of \cite{mogensen2020markov}.

\begin{proof}
	Note that we use Definitions \ref{def:unrollingExt} and 
	\ref{def:rollingExt} in this proof. One should also note that $B$ and $E$ 
	are disjoint as $E\subseteq C$. The following proof applies if $D_{t'}$ is 
	an unrolled version of $D$ and if $D$ is a rolled version of $D_{t'}$.
	Consider a $d$-connecting path in $D_{t'}$, say between $\nu_{s_0}^i \in 
	\bar{\nu}_{t-1}^A$ and $\nu_{s_{m+1}}^j \in {\nu}_{t}^B$ given 
	$\bar{\nu}_{t-1}^{B\cup C}\cup \nu_t^E$,
	
	\begin{align*}
	\nu_{s_0}^i \sim \nu_{s_1}^{k_1} \sim \ldots \sim \nu_{s_m}^{k_m} \sim 
	\nu_{s_{m+1}}^j
	\end{align*}
	
	\noindent where $\sim$ denotes an edge. We can assume that 
	$\nu_{s_{k+1}}^j$ is the only node in $\bar{\nu}_t^B$. We have $s_{m+1} = 
	t$ and the path must be nontrivial as it spans different lags. It must also 
	have a head at $\nu_{s_{m+1}}^j$ as otherwise there is a collider at a lag 
	$t''$, $t\leq t''$, which would close the path as $\deta(B) \cap E= 
	\emptyset$. The previous node must be at an earlier lag $s < t$ as 
	otherwise it would be in the conditioning set as $\pat(B)\subseteq E$. 
	There exists a walk in the rolled graph
	
	\begin{align*}
	i \sim k_1 \sim \ldots \sim k_m \sim j
	\end{align*}
	
	\noindent such that every edge has the same orientation as in the original 
	walk. We have $i\in A$ and $j\in B$. If any non-endpoint node is in $B$, 
	then it must correspond to a node at lag $s$, $s<t$, and therefore a 
	collider. This means that we can find a subwalk with a head at $B$ such 
	that no non-endpoint node is in $B$ and therefore this subwalk is also in 
	$D^B$. If $k_l$ is a noncollider, then $\nu_{s_l}^{k_l}$ is also a 
	noncollider and $\nu_{s_l}^{k_l} \notin \bar{\nu}_{t-1}^{B\cup C}\cup 
	\nu_t^E$. If $\nu_{s_l}^{k_l}$ is in time lag $t$, then $k_l$ is in $E$ as 
	$E$ is tail ancestral, and this cannot be. This means that $k_l \notin B 
	\cup C$. If $k_l$ is a collider, then $\nu_{s_l}^{k_l}$ is also a collider, 
	and $k_l \in \an^\texttt{+}(B \cup C)$ as $E\subseteq C$. From this walk, 
	we can find a $\delta$-connecting path.
\end{proof}

\begin{proof}[Proof of Proposition \ref{prop:markov}]
	We will use Lemma \ref{lem:extDeltasepToDsep} with $E = \emptyset$. This 
	satisfies the conditions when there are no contemporaneous effects. Assume 
	that $B$ is $\delta$-separated from $A$ given $C$ in $D$ and let $t = 
	1, \ldots, {t'}$. Let $\bar{D}_t$ be the unrolled version of $D$ on $t$ 
	time lags. Then $\nu_t^B$ is $d$-separated from $\overline{\nu}_{t-1}^A$ 
	given 
	$\overline{\nu}_{t-1}^{B\cup C}$ in $\bar{D}_t$ (Lemma 
	\ref{lem:extDeltasepToDsep}), and therefore also in $D_{t'}$ as $D_{t'}$ 
	restricted to time points $0$ to $t$ is a subgraph of $\bar{D}_t$. The 
	assumption in the proposition implies that $X_t^B$ is 
	conditionally independent 
	of $\overline{X}_{t-1}^A$ given $\overline{X}_{t-1}^{B\cup C}$. This holds 
	for any $t = 1,\ldots, {t'}$ and therefore means that $A$ is Granger 
	non-causal for $B$ given $ C$ until time ${t'}$.
\end{proof}

\begin{proof}[Proof of Theorem \ref{thm:markovContemp}]
	Let $0 < t \leq {t'}$ and consider the unrolled graph of $D$ on $t$ lags, 
	$\bar{D}_t$. The graph $\bar{D}_t$ is a supergraph of $D_t$, but they are 
	not necessarily equal. Assume there is a $d$-connecting path between 
	$v_{s_0}^{k_0}$ and $v_{t}^{k_l}$ given $\bar{\nu}_{t-1}^{B\cup C}$ such 
	that $s_0<t$, $k_0\in A$ and $k_l \in (\antv(B) \cap C) \cup B$. Using 
	arguments similar to those in the proof of Lemma 
	\ref{lem:extDeltasepToDsep}, we will argue that we can find a path between 
	$A$ and $(\antv(B) \cap C) \cup B$ in $D$ such that every collider is in 
	$\an^\texttt{+}(C \setminus \antv(B))$ and no non-collider is in $C 
	\setminus \antv(B)$ and such that there is a head at the final node. There 
	are no nodes at lag $t$ in the conditioning set and therefore there exists 
	an $l'$ such that consecutive nodes $v_{s_{l'}}^{k_{l'}}, 
	v_{s_{l'+1}}^{k_{l'+1}}, \ldots, v_{s_{l-1}}^{k_{l-1}}, v_{t}^{k_{l}}$ are 
	at lag $t$ and no other nodes on the original path are at lag $t$. All the 
	nodes $k_{l'}, k_{l'+1}, \ldots, k_{l-1}$ are in $\antv(B)$ and the subpath 
	from $v_{s_{l'-1}}^{k_{l'-1}}$ to $v_{t}^{k_{l'}}$ in $\bar{D}_t$ is 
	directed. We consider the subpath from $v_{s_0}^{k_0}$ to the first node 
	$v_{t}^{k_{l''}}$ such that  $k_{l''}\in (\antv(B) \cap C) \cup B$ and the 
	corresponding walk
	
	\begin{align*}
	k_0 \sim k_1 \sim \ldots \sim k_{l''}
	\end{align*}
	
	\noindent in $D$. We can reduce  this to a path from $k_0$ to $(\antv(B) 
	\cap C) \cup B$ with a head at the final node. This path is nontrivial as 
	$A\cap \antv(B) = \emptyset$. If it contains a non-endpoint node which is 
	in ${(\antv(B) \cap C) \cup B}$, then this node is a collider (as it 
	corresponds to a lag $s$ such that $s < t$) and we can choose a subpath 
	which has a head at the final node and is in $D^{(\antv(B) \cap C) \cup 
	B}$. Otherwise, the path itself is in $D^{(\antv(B) \cap C) \cup B}$. If 
	$k_i$ is a noncollider, then it is also a noncollider on the $d$-connecting 
	path. Therefore, if $s < t$, $\nu_{s}^{k_i} \notin \bar{\nu}_{t-1}^{B\cup 
	C}$ and $k_i \notin C\setminus \antv(B)$. On the other hand, if $s = t$ 
	then $k_i\notin C$ as $k_{l'''} \in \antv(B)$ for every node 
	$\nu_t^{k_{l'''}}$ on the subpath at lag $t$. If $k_i$ is a collider, then 
	$\nu_{s}^{k_i} \in \an^\texttt{+}(\bar{\nu}_{t-1}^{B\cup C})$, so $k_i\in 
	\an^\texttt{+}(B\cup C)$. We can then find a path, $\omega$, from $A$ to 
	$(\antv(B) \cap C) \cup B$ in $D^{(\antv(B) \cap C) \cup B}$ such that no 
	noncollider is in $C$ and every collider is in $\an^\texttt{+}(C)$.
	
	If every collider on $\omega$ is also in $\an^\texttt{+}(C\setminus 
	\antv(B))$, then this path is $\delta$-connecting given  $C\setminus 
	\antv(B)$. Otherwise, consider the subpath from $A$ to the first collider 
	which is not in $\an^\texttt{+}(C\setminus \antv(B))$. This collider is in 
	$\an^\texttt{+}(\antv(B) \cap C)$ and we concatenate this subpath with the 
	directed path from the collider to $\antv(B)\cap C$ if the collider is not 
	itself in $\antv(B)\cap C$. This path is $\delta$-connecting from $A$ to 
	$(\antv(B) \cap C) \cup B$ given $C\setminus \antv(B)$. In conclusion, if 
	$A \not\rightarrow_\delta (\antv(B) \cap C) \cup B \mid C \setminus 
	\antv(B)$, then $\bar{\nu}_{t-1}^A$ and $\nu_t^B$ are $d$-separated given 
	$\bar{\nu}_{t-1}^{B\cup C}$ in $\bar{D}_t$ and therefore also in 
	${D}_{t'}$. The result follows from the $d$-separation Markov property of 
	the unrolled graph.
\end{proof}

We will use the next lemma to show how $\delta$-separation in an extended local 
independence graph may also represent the assumptions needed for a mediation 
analysis. We say that a node set $E\subseteq V$ is \emph{tail ancestral} if 
there is no $i\in V\setminus E$ and $j \in E$ such that $i \tailedrightarrow j$.

\begin{corollary}
	Let $(D,D_{t'})$ be a proper pair and let $A,C \subseteq V\cup W$ and 
	$B\subseteq V$ be disjoint. Assume that there are no tailed edges out of 
	$B$, that $\pat(B)\subseteq E$ and that $E\subseteq C$. If $B$ is 
	$\delta$-separated from $A$ given $C$ in $D$, then $\nu_t^B$ is 
	$d$-separated from $\bar{\nu}_{t-1}^A$ given $\bar{\nu}_{t-1}^{B\cup C} 
	\cup \nu_t^{E}$ in $D_{t'}$.
	\label{cor:extDeltasepToDsep1}
\end{corollary}

\begin{proof}
	This follows directly from Lemma \ref{lem:extDeltasepToDsep}.
\end{proof}

\begin{corollary}
	Let $(D,D_{t'})$ be a proper pair and let $A,C \subseteq V\cup W$ and 
	$B\subseteq V$ be disjoint. Assume that there are no tailed edges into $B$. 
	If $B$ is $\delta$-separated from $A$ given $C$ in $D$, then $\nu_t^B$ is 
	$d$-separated from $\bar{\nu}_{t-1}^A$ given $\bar{\nu}_{t-1}^{B\cup C}$ in 
	$D_{t'}$.
	\label{cor:extDeltasepToDsep2}
\end{corollary}

\begin{proof}
	This follows from Lemma \ref{lem:extDeltasepToDsep} with $E = \emptyset$. 
	Alternatively, this follows from the proof of Theorem 
	\ref{thm:markovContemp} by noting that $\ant(B)=\emptyset$.
\end{proof}

\begin{proof}[Proof of Proposition \ref{contempAssump}]
	Fix $t>0$. We first note that when $B$ is $\delta$-separated from $A$ given 
	$C$ in $D$, then the same is true in the rolled version of $D_t$, 
	$\bar{D}$, as every edge in this graph is also in $D$. For $t_k > 0$, the 
	statements concerning A1 and A3 follow from Corollary 
	\ref{cor:extDeltasepToDsep1} using the pair $(\bar{D}, D_t)$. The statement 
	concerning A2 follows from Corollary \ref{cor:extDeltasepToDsep2}. For $t_k 
	= 0$, the statement follows directly from the causal graph as there are no 
	contemporaneous edges at $A^D$ or at $A^M$.
\end{proof}

\begin{remark}
	The proof of Theorem \ref{thm:markovContemp} uses that if $(\ant(B)\cap 
	C\cap V) \cup B$ is $\delta$-separated from $A$ given $C \setminus (\ant(B) 
	\cap V)$, then $\bar{\nu}_{t-1}^A$ and $\nu_t^B$ are $d$-separated given 
	$\bar{\nu}_{t-1}^C$ in the unrolled graph $\bar{D}_t$. The opposite 
	statement does not hold without further restrictions. As an example of 
	this, one can look at the graph $F_1 \tailedleftarrow F_2 \leftarrow F_3 
	\rightarrow F_4$ and sets $A = \{F_1\}$, $B = \{F_3\}$, $C=\emptyset$ where 
	$F_1$ corresponds to a baseline variable.
	
	If $A\cap \ant(B) = \emptyset$ and $A \not\rightarrow_\delta (\ant(B) \cap 
	C) \cup B \mid C \setminus \ant(B)$, then the result also holds. However, 
	this is a weaker formulation as this condition implies the condition in the 
	theorem. Moreover, in the rolled graph $D$ on nodes $\{A,B,C_1,C_2\}$ such 
	that $A$, $C_1$ and $C_2$ represent baseline variables and with edges $A 
	\tailedrightarrow C_1 \tailedleftarrow C_2 \tailedrightarrow B$ as well as 
	$C_1\tailedrightarrow B$, the condition of the theorem holds while the 
	above notion of separation does not hold.
	\label{rem:thm}
\end{remark}

\section{A continuous-time example}
\label{ssec:continuousTimeExample}

\begin{figure}
	\begin{subfigure}{.3\textwidth}
		\begin{tikzpicture}
		\begin{scope}[every node/.style={thick,draw=none}]
		\node (A) at (0,0) {$A$};
		\node (M) at (2,1) {$M$};
		\node (D) at (2,-1) {$D$};
		\node (L) at (4,2) {$L$};
		\node (U) at (4,0) {$U$};
		\end{scope}
		
		\begin{scope}[>={Stealth[black]},
		every node/.style={fill=white,circle},
		every edge/.style={draw=black,very thick}]
		\path [->] (A) edge (M);
		\path [->] (A) edge (D);
		\path [->] (M) edge[bend left = 10] (D);
		\path [->] (L) edge (M);
		\path [->] (L) edge (D);
		\path [->] (U) edge (L);
		\path [->] (U) edge (D);
		\end{scope}
		\end{tikzpicture}
	\end{subfigure}\hfill
	\begin{subfigure}{.3\textwidth}
		\begin{tikzpicture}
		\begin{scope}[every node/.style={thick,draw=none}]
		\node (Am) at (2,1) {$A_{M}$};
		\node (Ad) at (2,-1) {$A_{D}$};
		\node (M) at (4,1) {$M$};
		\node (D) at (4,-1) {$D$};
		\node (L) at (6,2) {$L$};
		\node (U) at (6,0) {$U$};
		\end{scope}
		
		\begin{scope}[>={Stealth[black]},
		every node/.style={fill=white,circle},
		every edge/.style={draw=black,very thick}]
		\path [->] (Ad) edge (D);
		\path [->] (Am) edge (M);	
		\path [->] (M) edge[bend left = 10] (D);
		\path [->] (L) edge (M);
		\path [->] (L) edge (D);
		\path [->] (U) edge (L);
		\path [->] (U) edge (D);
		\end{scope}
		\end{tikzpicture}
	\end{subfigure}
	\caption{Example graphs representing the linear Hawkes example. Left: graph 
	representing the observational distribution. Right: hypothetical 
	intervention. Process $U$ is unobserved.}
	\label{Figure:hawkes}
\end{figure}
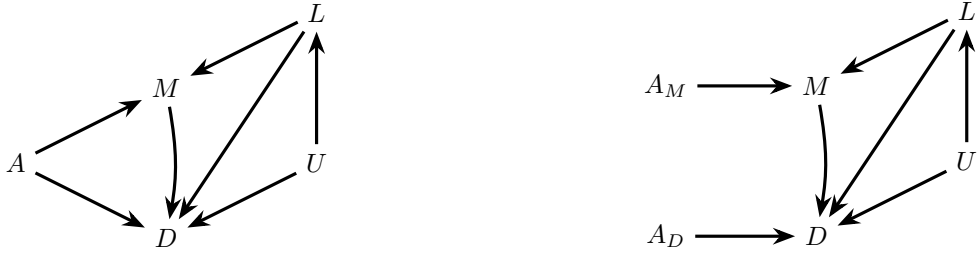

We now illustrate how local independence graphs may also represent sufficient 
assumptions for mediation analysis in continuous-time stochastic processes. For 
this purpose, we consider a multivariate \emph{linear Hawkes process} 
\citep{hawkes1971point}. We will first give some background and explain the 
result. In Subsection \ref{sssec:hawkesDetails}, we will then provide details 
of how direct and mediated effects can be defined and identified in this 
setting.

A linear Hawkes process is a multivariate point process evolving in continuous 
time (see also the example in Figure \ref{fig:hawkesExample}). We let $T_k^i$ 
denote the event times in the $i$'th process. We let $N = (N^1,\ldots,N^n)$ 
denote the corresponding counting processes, $N_t^i = \sum_{k} 
\mathds{1}_{T_k^i\leq t}$. The distribution of an $n$-dimensional point process 
is described by its conditional intensities. Heuristically, the conditional 
intensity of process $i$, $\lambda_t^i$, is the limiting probability of 
observing an event in process $i$ in the infinitesimal interval $(t, t+h]$ 
given the past of the process until time point $t$. A point process is a 
\emph{linear Hawkes process} if for each $i = 1,\ldots,n$ and $t \in 
\mathbb{R}$ the conditional intensity of process $i$, satisfies

$$
\lambda_t^i = \mu_i + \sum_{j=1}^n \int_{-\infty}^t g_{ij}(t-s) \mathrm{d}N_s^j 
=  \mu_i + \sum_{j=1}^n \sum_{T_k^j < t} g_{ij}(t - T_k^j)
$$

\noindent where $\mu_i$ are nonnegative constants and $g_{ij}$ are nonnegative 
functions. We see that the intensity of process $i$ at time $t$ has a 
contribution from each $j$-event occurring prior to time $t$ when $g_{ij}(t - 
T_k^j)$ is nonzero. A linear Hawkes process models recurrent events and there 
may therefore be multiple events in both treatment, mediation and outcome 
processes, see example data in Figure \ref{fig:hawkesExample} (left).

The connections between the coordinate processes are described by functions 
$g_{ij}$. We define a matrix $G$ such that $G_{ij} = \int_{-\infty}^\infty 
g_{ij}(t) \mathrm{d}t$ which will provide us with scalar parameters that we 
will identify in the mediation analysis. Clearly, $G_{ij}$ is a measure of the 
strength of the direct influence from process $j$ to process $i$ as $g_{ij}$ is 
a nonnegative function. Note that $G_{ij} = 0$ implies $g_{ij} = 0$ if $g_{ij}$ 
is continuous. We assume that the \emph{spectral radius} of $G$, that is, the 
largest absolute value of the eigenvalues of $G$, is strictly less than one. 
This also allows us to assume that the process is stationary. We assume that 
$G$ is \emph{normalised} (see Subsection \ref{sssec:hawkesDetails}).   We can 
define a graph from a Hawkes process such that $j\rightarrow i$ is in the graph 
if and only if $G_{ij} \neq 0$ and we say that this is the \emph{causal graph}. 
This is also a local independence graph in the sense that $\delta$-separation 
implies local independence \citep{didelez08, mogensen2020causal} 

The matrix $G$ has a useful interpretation which is most easily explained using 
the \emph{cluster representation} of the linear Hawkes process. We introduced 
the linear Hawkes process using its intensities above, however, one may give an 
equivalent definition as follows. For each $j$ we generate generation-0 events 
from a Poisson process with rate $\mu_j$. Each generation-0 event creates a 
\emph{cluster} of future events recursively. From a generation-0 event of type 
$j$ at time $s$, we generate for all $i = 1,\ldots,n$ generation-1 events at 
times $t>s$ from an inhomogeneous Poisson process with rate $g_{ij}(t-s)$ for 
$t > s$. From generation-1 events, generation-2 events are generated 
analogously and so forth. This generates a single \emph{cluster}. By combining 
the event types and event times of all clusters, we obtain a linear Hawkes 
process, though the cluster structure itself is unknown.

We see that there is a simple parent-child relation in the data generation, 
even though this is not observed in the data, see Figure 
\ref{fig:hawkesExample} where parent and child events are joined by a line 
segment (the parent event always occurs before the child event). The number of 
expected direct $i$-descendants of a $j$-event is simply $G_{ij}$ 
(\emph{direct} means that if the parent event is in the $n$'th generation in a 
cluster, then the child event is in the $(n+1)$'th generation, see Figure 
\ref{fig:hawkesExample} for an illustration). If we define $R = (I - G)^{-1} = 
\sum_{k = 0}^\infty G^k$, then $R_{ij}$ is the total expected number of 
$i$-events in a cluster rooted at a $j$-event \citep{jovanovic2015cumulants}. 
Note that this matrix is well-defined due to the assumption on the spectral 
radius on $G$. We will say that a linear Hawkes process is \emph{causal} if the 
cluster generated by an injected (interventional) event has the same 
distribution as an intrinsic (noninterventional) event 
\citep{mogensen2022equality}.

We now turn to the specific model in Figure \ref{Figure:hawkes} with a 
treatment process, $A$, a mediator process, $M$, an outcome process $D$, a 
`covariate' process, $L$, and an unobserved confounder process, $U$. Note that, 
in contrast to previous sections, $N$ denotes the multivariate process and the 
outcome process is denoted by $D$. In line with the above discussion, the 
mediation analysis consists of identifying the number of expected outcome 
events created by injection of a treatment ($A$) event mediated through $M$ and 
directly. Using the causal assumption, this amounts to identifying the relevant 
entries of $G$. We can imagine manipulating the system such that we can inject 
an event in the treatment process which only propagates directly to the 
treatment, or only to the mediator (see Figure \ref{Figure:hawkes}, right). The 
expected number of outcome events in this manipulated cluster are the direct 
effect and the mediated effect, respectively.

The parameter $R_{DA}$ is the expected number of $D$-events on cluster rooted 
at an $A$-event. We have $R = \sum_{i=0}^\infty G^i$ and it follows that

$$
R_{DA} = G_{DA} + G_{DM}G_{MA},
$$

\noindent using the fact that $G$ is normalised and therefore has zeros on the 
diagonal. We see that the parameter $G_{DA}$ is the expected number of 
$D$-events directly from an $A$-event. On the other hand, $G_{DM}G_{MA}$ is the 
expected number of $D$-events on an $A$-cluster that are mediated through an 
$M$-event. This means that if we were to intervene and inject an event into 
$A^M$ only, this would on average create $G_{DM}G_{MA}$ $D$-events. On the 
other hand, if we were to inject an event into $A^D$ only, this would on 
average create $G_{DA}$ $D$-events. Therefore, we may think of $R_{DA}$ as a 
total effect from $A$ to $D$ while $G_{DA}$ is the direct effect from $A$ to 
$D$ and  $G_{DM}G_{MA}$ is the effect from $A$ to $D$ mediated by $M$. In the 
next subsection, we show that both of these quantities are identified from the 
observational distribution.

\subsection{Details}
\label{sssec:hawkesDetails}

We define a matrix, $C$, such that

$$
C_{ij}dt = \int_{-\infty}^\infty E(dN_t^idN_{t+\tau}^j) - 
E(dN_t^id)E(dN_{t+\tau}^j)\mathrm{d}\tau
$$

\noindent and we say that $C$ is the \emph{integrated covariance}. We will 
argue that the mediated and direct effects from $A$ to $D$ are in fact 
identified from the observed integrated covariance of the linear Hawkes process 
model corresponding to Figure \ref{Figure:hawkes} (left).
We assume that the parameter matrix $G$ is \emph{normalised}, i.e., has zeros 
on the diagonal. A normalised representation can always be obtained and the 
normalised parameters have a simple interpretation in that a normalised 
parameter $G_{ij}$ represents the expected number of direct $i$-events from an 
$j$-event counting all subsequent direct self-events, that is, all $i$-events 
on subclusters of the type $j - i - i - \ldots - i$ 
\citep{mogensen2022equality}.

\begin{figure}
	
	\begin{subfigure}{0.45\textwidth}
		\centering
		\includegraphics[width=\textwidth]{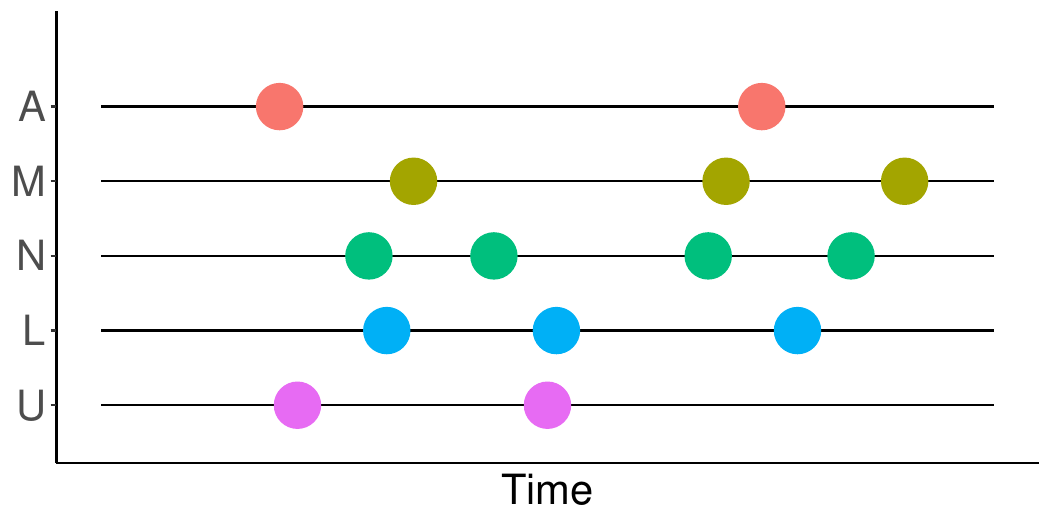}	
		\end{subfigure}\hspace{.05\textwidth}%
	\begin{subfigure}{0.45\textwidth}
		\centering
		\includegraphics[width=\textwidth]{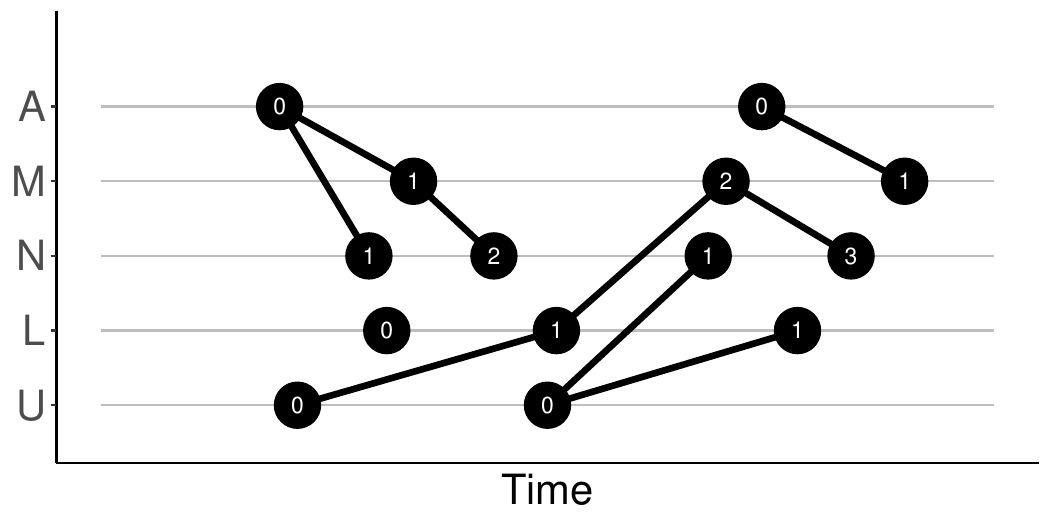}
	\end{subfigure}
	\caption{\label{fig:hawkesExample} Example of data from the linear Hawkes 
	example. The observed data contain pairs $(t_k,i_k)$ where $t_k$ is a time 
	point and $i_k$ is the coordinate process in which the event occurs. On the 
	left, vertical placement and colour illustrate the coordinate process. A 
	linear Hawkes process can be generated as an offspring process. 
	Generation-0 events are generated. A generation-0 event at time $t$ of type 
	$j$ creates generation-1 events of type $i$ corresponding to an 
	inhomogeneous Poisson process with rate $g_{ij}$ started at time $t$. 
	Analogously, a generation-1 event at time $t$ of type $j$ creates 
	generation-2 events of type $i$ corresponding to an inhomogeneous Poisson 
	process with rate $g_{ij}$ started at time $t$, and so forth. On the right, 
	the cluster structure is also plotted, though this is not observed in the 
	data (numbers indicate generation within the cluster).}
\end{figure}

\begin{proposition}
	Consider a stationary linear Hawkes process such that the spectral radius 
	of $G$ is strictly less than one and such that $E(\lambda_t^i) > 0$ for 
	each $i$. When the graph in Figure \ref{Figure:hawkes} (left) represents 
	the model and processes $O = \{A,M,D,L\}$ are observed, then both the 
	(normalised) direct effect from $A$ to $D$ and the (normalised) effect from 
	$A$ to $D$ mediated by $M$ are identified from the observed integrated 
	covariance.
\end{proposition}

\begin{proof}
	We can write the observable part of the integrated covariance as
	
	$$
	C_{OO} = (I - G_{OO})^{-1}\Theta (I - G_{OO})^{-T}
	$$
	
	\noindent where $\Theta$ is a positive definite matrix such that 
	off-diagonal entries are zero except for $\Theta_{34}$ and $\Theta_{43}$ 
	\citep{mogensen2022equality}. One can show that $R_{OO}=(I - G_{OO})^{-1}$. 
	Theory on linear structural equation models give that $G_{MA}$, $G_{DA}$ 
	and $G_{DM}$ are all generically identified from $C_{OO}$, that is, except 
	for choices of parameters $G_{OO}$ and $\Theta_{OO}$ of measure zero 
	\citep{weihs2018determinantal}. We show directly that they are in fact 
	always identified using that the diagonal of $\Theta$ is strictly positive. 
	
	The matrix equation above is $4 \times 4$. Note that $R_{ji} = 0$ if there 
	is no directed path from $i$ to $j$ and $i\neq j$. The matrix $I - G$ is 
	upper triangular (after re-arranging the rows and columns) and therefore 
	$R$ is as well. Using $I = R(I-G)$ gives that $R$ must have ones on the 
	diagonal. The following is a sketch of how to solve the equations to 
	sequentially identify the parameters we need. From the $C_{AA}$-equation, 
	we can identify $\Theta_{AA}$. From the $C_{MA}$-equation, we can then 
	identify $R_{MA}$ and from the $C_{DA}$-equation we can identify $R_{DA}$. 
	From $C_{LL}$, we identify $\Theta_{LL}$. From the $C_{LM}$-equation we 
	identify $\Theta_{LL}R_{ML}$ and then $R_{ML}$. From the $C_{MM}$-equation 
	we identify $\Theta_{MM}$. From the $C_{DM}$-equation we identify $R_{DM}$. 
	Then finally, we can identify $G_{MA}, G_{DA}, G_{DM}$ from the inverting 
	the matrix $R_{\bar{O}\bar{O}}$ where $\bar{O} = \{A,M,D\}$. This matrix is 
	invertible since it is triangular and has ones on its main diagonal. 
\end{proof}

\bibliography{papers6}

\end{document}